\documentclass[11pt]{article}
\usepackage[T1]{fontenc}
\usepackage{geometry} 
\usepackage{graphicx}
\usepackage{amsthm}
\usepackage{amsmath,amssymb,amsfonts,dsfont,mathtools}
\usepackage{enumitem}
\usepackage{relsize}
\usepackage[margin=1cm]{caption}
\usepackage{wrapfig}
\usepackage{xcolor}
\usepackage{frame,color}
\usepackage{environ,subfig}
\usepackage{hyperref}
\usepackage[square,sort,comma,numbers]{natbib}
\usepackage{xspace}
\usepackage{framed}
\usepackage{comment}
\usepackage{changepage}
\usepackage[boxed]{algorithm2e}
\usepackage{algorithmic}
\usepackage{mathrsfs}
\usepackage{tabu} 
\usepackage{soul} 

\usepackage[capitalize, nameinlink]{cleveref}

\makeatletter
\renewenvironment{framed}{%
 \def\FrameCommand##1{\hskip\@totalleftmargin
 \fboxsep=\FrameSep\fbox{##1}
     \hskip-\linewidth \hskip-\@totalleftmargin \hskip\columnwidth}%
 \MakeFramed {\advance\hsize-\width
   \@totalleftmargin\z@ \linewidth\hsize
   \@setminipage}}%
 {\par\unskip\endMakeFramed}
\makeatother

\newtheorem{theorem}{Theorem}[section]
\newtheorem{lemma}[theorem]{Lemma}
\newtheorem{meta-theorem}[theorem]{Meta-Theorem}
\newtheorem{claim}[theorem]{Claim}
\newtheorem{remark}[theorem]{Remark}

\definecolor{darkgreen}{rgb}{0,0.5,0}
\usepackage{hyperref}
\hypersetup{
    unicode=false,          
    colorlinks=true,        
    linkcolor=red,          
    citecolor=darkgreen,        
    filecolor=magenta,      
    urlcolor=cyan           
}

\crefname{theorem}{Theorem}{Theorems}
\Crefname{lemma}{Lemma}{Lemmas}
\Crefname{figure}{Figure}{Figures}

\newcommand{\ignore}[1]{}

\newcommand{\Expect}{\operatorname{E}}

\newcommand{\Prob}{\operatorname{Pr}}

\newcommand{\paren}[1]{\mathopen{}\left( #1 \right)\mathclose{}}

\newcommand{\poly}{{\operatorname{poly}}}
\newcommand{\polylog}{\poly \log}

\newcommand{\outdeg}{\operatorname{outdeg}}
\newcommand{\dist}{\operatorname{dist}}

\newcommand{\sparse}{\mathsf{sp}}
\newcommand{\bad}{\mathsf{bad}}

\newcommand{\Nout}{N^{\operatorname{out}}}
\newcommand{\Ninc}{{N}} 

\newcommand{\Det}{\mathsf{Det}}
\newcommand{\Detd}{\Det_{\scriptscriptstyle d}}

\newcommand{\ID}{\operatorname{ID}}
\newcommand{\LOCAL}{\mathsf{LOCAL}}
\newcommand{\LCA}{\mathsf{LCA}}

\newcommand{\CONGEST}{\mathsf{CONGEST}}
\newcommand{\MPC}{\mathsf{MPC}}

\newcommand{\CLIQUE}{\mathsf{CONGESTED}\text{-}\mathsf{CLIQUE}}

\newcommand{\trial}{{\sf ColorBidding}}

\newcommand{\speedupalgo}{{Opportunistic Speed-up}}  

\newcommand{\badvertex}{{\bf Bad}}
\newcommand{\bbbb}{{\sf{Bad}}}
\newcommand{\Good}{\sf{Good}}

\newcommand{\Cdecomp}{4}
\newcommand{\Ctries}{3}
\newcommand{\Csamples}{\beta}

\newcommand{\randcolor}{\mathscr{R}}

\newcommand{\DeltaSp}{\Delta_{\ast}}
\newcommand{\NSp}{N_{\ast}}

\newcommand{\Algo}{\mathcal{A}}

\newcommand{\binomial}{\operatorname{Binomial}}

\newcommand{\Invariant}{\mathcal{I}}

\newcommand{\InSize}{\ell_{\operatorname{in}}}
\newcommand{\OutSize}{\ell_{\operatorname{out}}}

\newcommand{\fyi}[1]{{\color{blue} \emph{#1}}}  

    \newcommand{\Evoverloaded}{E_v^{\mathrm{overload}}}
    \newcommand{\Evlucky}{E_v^{\mathrm{lucky}}}
    \newcommand{\Evrich}{E_v^{\mathrm{rich}}}
    \newcommand{\Evlazy}{E_v^{\mathrm{lazy}}}

\title{The Complexity of $(\Delta+1)$ Coloring  in \\ Congested Clique, Massively Parallel Computation, \\ and Centralized Local Computation}
\author{
Yi-Jun Chang \\
\small U. Michigan \\
\small cyijun@umich.edu
\and
Manuela Fischer\\
\small ETH Zurich \\
\small manuela.fischer@inf.ethz.ch
\and
Mohsen Ghaffari\\
\small ETH Zurich \\
 \small ghaffari@inf.ethz.ch
\and
Jara Uitto \\
\small ETH Zurich \& U. Freiburg \\
\small jara.uitto@inf.ethz.ch
\and
Yufan Zheng \\
\small U. Michigan \\
\small lwins.lights@gmail.com
 }

\begin{document}
\date{}
\maketitle
\thispagestyle{empty}
\setcounter{page}{0}

\begin{abstract}
In this paper, we present new randomized algorithms that improve the complexity of the classic $(\Delta+1)$-coloring problem, and its generalization $(\Delta+1)$-list-coloring, in three well-studied models of distributed, parallel, and centralized computation:
\begin{description}
\item[Distributed Congested Clique:]  We present an $O(1)$-round randomized algorithm for $(\Delta+1)$-list coloring in the congested clique model of distributed computing. This settles the asymptotic complexity of this problem. It moreover improves upon the $O(\log^\ast \Delta)$-round randomized algorithms of Parter and Su [DISC'18] and $O((\log\log \Delta)\cdot \log^\ast \Delta)$-round randomized algorithm of Parter [ICALP'18].

\item[Massively Parallel Computation:] 
We present a $(\Delta+1)$-list coloring algorithm with round complexity $O(\sqrt{\log\log n})$ in the Massively Parallel Computation ($\MPC$) model with strongly sublinear memory per machine. This algorithm uses a memory of $O(n^{\alpha})$ per machine, for any desirable constant $\alpha>0$, and a total memory of $\widetilde{O}(m)$, where $m$ is the size of the graph. Notably, this is the first coloring algorithm with sublogarithmic round complexity, in the sublinear memory regime of $\MPC$. For the quasilinear memory regime of $\MPC$, an $O(1)$-round algorithm was given very recently by Assadi et al. [SODA'19].

\item[Centralized Local Computation:] We show that $(\Delta+1)$-list coloring can be solved with $\Delta^{O(1)} \cdot O(\log n)$ query complexity, in the centralized local computation model. The previous state-of-the-art for $(\Delta+1)$-list coloring in the centralized local computation model are based on  simulation of known $\LOCAL$ algorithms. The deterministic $O(\sqrt{\Delta} \polylog \Delta + \log^\ast n)$-round $\LOCAL$ algorithm of Fraigniaud et al. [FOCS'16] can be implemented in the centralized local computation model with query complexity $\Delta^{O(\sqrt{\Delta} \polylog \Delta)} \cdot O(\log^\ast n)$; the randomized $O(\log^\ast \Delta) +2^{O(\sqrt{\log \log n})}$-round $\LOCAL$ algorithm of Chang et al. [STOC'18] can be implemented in the centralized local computation model with query complexity $\Delta^{O(\log^\ast  \Delta)} \cdot O(\log n)$. 
\end{description}
\end{abstract}
\newpage

\section{Introduction, Related Work, and Our Results}
In this paper, we present improved randomized algorithms for vertex coloring in three models of distributed, parallel, and centralized computation: the {\em congested clique} model of distributed computing, the {\em massively parallel computation} model, and the {\em centralized local computation} model. We next overview these results in three different subsections, while putting them in the context of the state of the art. The next section provides a technical overview of the known algorithmic tools as well as the novel ingredients that lead to our results.

\paragraph{$(\Delta + 1)$-coloring and $(\Delta + 1)$-list Coloring.} Our focus is on the standard $\Delta+1$ vertex coloring problem, where $\Delta$ denotes the maximum degree in the graph. All our results work for the generalization of the problem to $(\Delta + 1)$-list coloring problem, defined as follows: each vertex $v$ in the graph $G=(V, E)$ is initially equipped with a set of colors $\Psi(v)$ such that $|\Psi(v)| = \Delta + 1$. The goal is to find a proper vertex coloring where each vertex $v\in V$ is assigned a color in  $\Psi(v)$ such that no two adjacent vertices are colored the same.

\subsection{Congested Clique Model of Distributed Computing}
\paragraph{Models of Distributed Computation.}
There are three major models for distributed graph algorithms, namely $\LOCAL$, $\CONGEST$, and $\CLIQUE$. In the $\LOCAL$ model~\cite{Linial92,Peleg00}, the input graph $G=(V,E)$ is identical to the communication network and each $v\in V$ hosts a processor that initially knows $\deg(v)$, a unique $\Theta(\log n)$-bit $\ID(v)$,
and global graph parameters $n = |V|$ and $\Delta =\max_{v\in V} \deg(v)$. Each processor
is allowed unbounded computation and has access to a stream of private random bits. \emph{Time} is partitioned into synchronized \emph{rounds} of communication, in which each processor sends one unbounded message to each neighbor.
At the end of the algorithm, each $v$ declares its output label, e.g., its own color. The $\CONGEST$ model~\cite{Peleg00} is a variant of  $\LOCAL$ where there is an $O(\log n)$-bit message size constraint. The $\CLIQUE$ model, introduced in \cite{lotker2005CongestK},  is a variant of $\CONGEST$ that allows all-to-all communication: Each vertex initially knows its adjacent edges of the input graph $G=(V, E)$. In each round, each vertex is allowed to transmit $n-1$ many $O(\log n)$-bit messages, one addressed to each other vertex.

In this paper, our new distributed result is an improvement for coloring in the $\CLIQUE$ model. It is worth noting that the $\CLIQUE$ model has been receiving extensive attention recently, see e.g., \cite{Patt-Shamir2011sorting, dolev2012tri, berns2012super, Lenzen2013, Drucker:congestedK, Danupon-paths, hegeman2014near, hegeman2015lessons, Censor2016, Hegeman:MST, becker_et_al:transshipment, Gall2016FurtherAA, Censor-Hillel:SparseMatrix, Ghaffari16, Ghaffari2017, ghaffari2018improved, ParterS18, Parter18, DBLP:conf/csr/BarenboimK18}.


\paragraph{State of the Art for Coloring in $\LOCAL$ and $\CONGEST$.} Most prior works on distributed coloring focus on the $\LOCAL$ model. The current state-of-the-art randomized upper bound for the $(\Delta+1)$-list coloring problem is  $O(\log^\ast \Delta) + O(\Detd(\polylog n)) = O(\Detd(\polylog n))$ of~\cite{ChangLP18} (which builds upon the techniques of~\cite{HarrisSS18}), 
where $\Detd(n') = 2^{O(\sqrt{\log \log n'})}$ is the deterministic complexity
of $(\deg+1)$-list coloring on $n'$-vertex graphs~\cite{PanconesiS96}.
In the $(\deg+1)$-list coloring problem, each $v$ has a palette of size $\deg(v)+1$. This algorithm follows the {\em graph shattering} framework~\cite{BEPS16, Ghaffari16}. The pre-shattering phase takes $O(\log^\ast \Delta)$ rounds. After that, the remaining uncolored vertices form connected components of size $O(\poly \log n)$. The post-shattering phase then applies a $(\deg+1)$-list coloring deterministic algorithm to color all these vertices.

\paragraph{State of the Art for Coloring in $\CLIQUE$.} 
Hegeman and Pemmaraju~\cite{hegeman2015lessons} gave algorithms for $O(\Delta)$-coloring in the $\CLIQUE$ model, which run in $O(1)$ rounds if $\Delta\geq \Theta(\log^4 n )$ and in $O(\log\log n)$ rounds otherwise. It is worth noting that $O(\Delta)$ coloring is a significantly more relaxed problem in comparison to $\Delta+1$ coloring. For instance, we have long known a very simple $O(\Delta)$-coloring algorithm in $\LOCAL$-model algorithm with round complexity $2^{O(\sqrt{\log \log n})}$~\cite{BEPS16}, but only recently such a round complexity was achieved for $\Delta+1$ coloring ~\cite{ChangLP18,HarrisSS18}.    

Our focus is on the much more stringent $\Delta+1$ coloring problem. For this problem, the $\LOCAL$ model algorithms of~\cite{ChangLP18,HarrisSS18} need messages of $O(\Delta^2 \log n)$ bits, and thus do not extend to $\CONGEST$ or $\CLIQUE$. For $\CLIQUE$ model, the main challenge is when $\Delta > \sqrt{n}$, as otherwise, one can simulate the algorithm of~\cite{ChangLP18} by leveraging the all-to-all communication in $\CLIQUE$ which means each vertex in each round is capable of communicating $O(n \log n)$ bits of information. Parter~\cite{Parter18} designed the first sublogarithmic-time $(\Delta+1)$ coloring algorithm for $\CLIQUE$, which runs in $O(\log \log \Delta \log^\ast \Delta)$ rounds. The algorithm of~\cite{Parter18} is able to reduce the maximum degree to $O(\sqrt{n})$ in $O(\log \log \Delta)$ iterations, and each iteration invokes the algorithm of~\cite{ChangLP18} on instances of maximum degree $O(\sqrt{n})$. Once the maximum degree is $O(\sqrt{n})$, the algorithm of~\cite{ChangLP18} can be implemented in $O(\log^\ast \Delta)$ rounds in $\CLIQUE$. Subsequent to~\cite{Parter18}, the upper bound was improved to $O(\log^\ast \Delta)$ in~\cite{ParterS18}. Parter and Su~\cite{ParterS18} observed that the algorithm of~\cite{Parter18} only takes $O(1)$ iterations if we only need to reduce the degree to $n^{1/2 + \epsilon}$, for some constant $\epsilon > 0$, and they achieved this by modifying the internal details of~\cite{ChangLP18} to reduce the required message size to $O(\Delta^{8/5} \log n)$. 

\paragraph{Our Result.}
For the $\CLIQUE$ model, we present a new algorithm for $(\Delta+1)$-list coloring in the randomized congested clique
model running in $O(1)$ rounds. This improving on the previous best known $O(\log^\ast \Delta)$-round algorithm of Parter and Su~\cite{ParterS18} and settles the asymptotic complexity of the problem.

\begin{theorem}\label{thm-congested-clique-main}
    There is an $O(1)$-round algorithm that solves the $(\Delta+1)$-list coloring problem in $\CLIQUE$, with success probability $1 - 1/\poly(n)$.
\end{theorem}

The proof is presented in two parts: If $\Delta \geq \log^{4.1} n$, the algorithm of \Cref{thm:largedegree} solves the $(\Delta+1)$-list coloring problem in $O(1)$ rounds; otherwise, the algorithm of \Cref{thm:smalldegree} solves the problem in $O(1)$ rounds.

\subsection{Massively Parallel Computation}

\paragraph{Model.} The Massively Parallel Computation (MPC) model was introduced by Karloff et al.~\cite{KarloffSV10}, as a theoretical abstraction for practical large-scale parallel processing settings such as MapReduce~\cite{dg04}, Hadoop~\cite{White:2012}, Spark~\cite{ZahariaCFSS10}, and Dryad~\cite{Isard:2007}, and it has been receiving increasing more attention over the past few years~\cite{KarloffSV10, goodrich2011sorting, LattanziMSV11, Beame13, Andoni:2014, Beame14, hegeman2015lessons, AhnGuha15,Roughgarden16,Im17,czumaj2017round, assadi2017simple,assadi2017coresets,ghaffari2018improved,harvey2018greedy,brandt2018breaking,assadi2018massively,boroujeni2018approximating,Andoni2018}. In the $\MPC$ model, the system consists of a number of machines, each with $S$ bits of memory, which can communicate with each other in synchronous rounds through a complete communication network. Per round, each machine can send or receive at most $S$ bits in total. Moreover, it can perform some $\poly(S)$ computation, given the information that it has. In the case of graph problems, we assume that the graph $G$ is partitioned among the machines using a simple and globally known hash function such that each machine holds at most $S$ bits, and moreover, for each vertex or potential edge of the graph, the hash function determines which machines hold that vertex or edge. Thus, the number of machines is $\Omega(m/S)$ and ideally not too much higher, where $m$ denotes the number of edges. At the end, each machine should know the output of the vertices that it holds, e.g., their color. 

\paragraph{State of the Art for Coloring.}
The $\CLIQUE$ algorithms discussed above can be used to obtain $\MPC$ algorithms with the same asymptotic round complexity if machines have memory of $S=\Omega(n\log n)$ bits. In particular, the work of Parter and Su~\cite{ParterS18} leads to an $O(\log^* \Delta)$-round $\MPC$ algorithm for machines with $S=\Omega(n\log n)$ bits. However, this $\MPC$ algorithm would have two drawbacks: (A) it uses $\Omega(n^2 \log n)$ global memory, and thus would require $(n^2 \log n)/S$ machines, which may be significantly larger than $\tilde{O}(m)/S$. This is basically because the algorithm makes each vertex of the graph learn some $\widetilde{\Theta}(n)$ bits of information. (B) It is limited to machines with $S=\Omega(n\log n)$ memory, and it does not extend to the machines with strongly sublinear memory, which is gaining more attention recently due to the increase in the size of graphs. We note that for the regime of machines with super-linear memory, very recently, Assadi, Chen, and Khanna~\cite{AssadiCK18} gave an $O(1)$-round algorithm which uses only $O(n\log^3 n)$ global memory.\footnote{Here ``global memory'' refers to the memory used for communication. Of course we still need $\tilde{O}(m)$ memory to store the graph.} However, this algorithm also relies heavily on $S=\Omega(n\log^3 n)$ memory per machine and cannot be run with weaker machines that have strongly sublinear memory.  

\paragraph{Our Result.} We provide the first sublogarithmic-time algorithm for $(\Delta+1)$ coloring and $(\Delta+1)$-list coloring in the $\MPC$ model with strongly sublinear memory per machine:

\begin{theorem}\label{thm:MPCCol}
There is an $\MPC$ algorithm that, in $O(\log^*\Delta+\sqrt{\log\log n}) = O(\sqrt{\log\log n})$ rounds, w.h.p.\ computes a $(\Delta+1)$ list-coloring of an $n$-vertex graph with $m$ edges and maximum degree $\Delta$ and that uses $O(n^{\alpha})$ memory per machine, for an arbitrary constant $\alpha>0$, as well as a total memory of $\widetilde{O}(m)$. 
\end{theorem}
The proof is presented in \Cref{subsec:MPC}.

\subsection{Centralized LOCAL Computation}
\paragraph{Model.} This Local Computation Algorithms (LCA) model is a centralized model of computation that was introduced in~\cite{RubinfeldTVX11}; an algorithm in this model is usually called an $\LCA$.
In this model, there is a graph $G = (V, E)$ where the algorithm is allowed to make the following queries:
\begin{description}
    \item[Degree Query:] Given $\ID(v)$, the oracle returns $\deg(v)$. 
    \item[Neighbor Query:]  Given $\ID(v)$ and an index $i \in [1,\Delta]$,  if $\deg(v) \leq i$, the oracle returns $\ID(u)$, 
    where $u$ is the $i$th neighbor of $v$; otherwise, the oracle returns $\bot$.
\end{description}
It is sometimes convenient to assume that there is a query that returns the list of all neighbors of $v$. This query can be implemented using one degree query and $\deg(v)$ neighbor queries.
For  randomized algorithms, we assume that there is an oracle that given $\ID(v)$ returns an infinite-length random sequence associated with the vertex $v$.
Similarly, for problems with input labels (e.g., the color lists in the list coloring problem), the input label of a vertex $v$ can be accessed given $\ID(v)$. Given a distributed problem $\mathcal{P}$, an $\LCA$ $\Algo$ accomplishes the following. Given $\ID(v)$, the algorithm $\Algo$ returns $\Algo(v) =$ the output of $v$, after making a small number of queries. It is required that the output of $\Algo$ at different vertices are consistent with one legal solution of $\mathcal{P}$. 

The complexity measure for an $\LCA$ is the number of queries.
It is well-known~\cite{ParnasR07} that any $\tau$-round $\LOCAL$ algorithm $\Algo$ can be transformed into an $\LCA$  $\Algo'$ with query complexity $\Delta^{\tau}$. The  $\LCA$  $\Algo'$ simply simulates the $\LOCAL$ algorithm $\Algo$ by querying all radius-$\tau$ neighborhood of the given vertex $v$.
See~\cite{levi2017centralized} for a recent survey about the state-of-the-art in the centralized local model.

\paragraph{State of the Art $\LCA$ for Coloring} The previous state-of-the-art for $(\Delta+1)$-list coloring in the centralized local computation model are based on  simulation of known $\LOCAL$ algorithms. The deterministic $O(\sqrt{\Delta} \polylog \Delta + \log^\ast n)$-round $\LOCAL$ algorithm of~\cite{FraigniaudHK16, barenboim2018locally}\footnote{Precisely, the complexity is $O(\sqrt{\Delta} \log^{2.5} \Delta + \log^\ast n)$ in~\cite{FraigniaudHK16}, and this has been later improved to $O(\sqrt{\Delta \log \Delta} \log^\ast \Delta + \log^\ast n)$ in~\cite{barenboim2018locally}.} can be implemented in the centralized local computation model with query complexity $\Delta^{O(\sqrt{\Delta} \polylog \Delta)} \cdot O(\log^\ast n)$; the randomized $O(\log^\ast \Delta) +2^{O(\sqrt{\log \log n})}$-round $\LOCAL$ algorithm of~\cite{ChangLP18} can be implemented in the centralized local computation model with query complexity $\Delta^{O(\log^\ast  \Delta)} \cdot O(\log n)$.

\paragraph{Our Result.} We show that $(\Delta+1)$-list coloring can be solved with $\Delta^{O(1)} \cdot O(\log n)$ query complexity. 
Note that $\Delta^{O(1)} \cdot O(\log n)$ matches a ``natural barrier'' for randomized algorithms based on the graph shattering framework, as each connected component in the post-shattering phase has this size $\Delta^{O(1)} \cdot O(\log n)$.

\begin{theorem}\label{thm:lca-main}
    There is an centralized local computation algorithm that solves the $(\Delta+1)$-list coloring problem with query complexity $\Delta^{O(1)} \cdot O(\log n)$, with success probability $1 - 1/\poly(n)$.
\end{theorem}
 The proof is presented in  \Cref{sect-implement-color-bidding}.


\section{Technical Overview: Tools and New Ingredients}\label{sec:prelim}
In this section, we first review some of the known technical tools that we will use in our algorithms, and then we overview the two new technical ingredients that lead to our improved results (in combination with the known tools).

\paragraph{Notes and Notations.}
When talking about randomized algorithms, we require the algorithm to succeed \emph{with high probability (w.h.p.)}, i.e., to have success probability at least $1 - 1/\poly(n)$. For each vertex $v$, we write $N(v)$ to denote the set of neighbors of $v$. If there is an edge orientation, $\Nout(v)$ refers to the set of out-neighbors of $v$.
We write $N^k(v) = \{ u \in V \ | \ \dist(u,v) \leq k\}$.
We use subscript to indicate the graph $G$ under consideration, e.g., $N_G(v)$ or $\Nout_G(v)$.
In the course of our algorithms, we slightly abuse the notation to also use $\Psi(v)$ to denote the set of {\em available colors} of $v$. i.e.,  the subset of $\Psi(v)$ that excludes the colors already taken by its neighbors in $N(v)$. The number of {\em excess colors} at a vertex is the number of available colors minus the number of uncolored neighbors. Moreover, we make an assumption that each color can be represented using $O(\log n)$ bits. This is without loss of generality (in all of the models under consideration in our paper), since otherwise we can hash the colors down to this magnitude, as we allow a failure probability of $1/\poly(n)$ for randomized algorithms.


\subsection{Tools}

\paragraph{Lenzen's Routing.}
The routing algorithm of Lenzen~\cite{Lenzen2013} for $\CLIQUE$ allows us to deliver all messages in $O(1)$ rounds, as long as each vertex $v$ is the source and the
destination of at most $O(n)$ messages.
This is a very useful (and frequently used) communication primitive for designing $\CLIQUE$ algorithms.

\begin{lemma}[Lenzen's Routing]\label{lem:routing}
Consider a graph $G=(V,E)$ and a set of point-to-point routing requests, each given by the $\ID$s of the corresponding source-destination pair. As long as each vertex $v$ is the source and the destination of at most $O(n)$ messages, namely $O(n \log n)$ bits of information,  we can deliver all messages in $O(1)$ rounds in the $\CLIQUE$ model.
\end{lemma}

\paragraph{The Shattering Framework.}
Our algorithm follows the {\em graph shattering} framework~\cite{BEPS16}, which first performs some randomized process (known as \emph{pre-shattering}) to solve ``most'' of the problem, and then performs some clean-up steps (known as \emph{post-shattering}) to solve the remaining part of the problem. Typically, the remaining graph is simpler in the sense of having small components and having a small number of edges. Roughly speaking, at each step of the algorithm, we specify an invariant that all vertices must satisfy in order to continue to participate. Those {\em bad vertices} that violate the invariant are removed from consideration, and postponed to the post-shattering phase. We argue that the bad vertices form connected components of size $\Delta^{O(1)} \cdot O(\log n)$ with probability $1 - 1/\poly(n)$; we use this in designing $\LCA$.
Also, the total number of edges induced by the bad vertices is $O(n)$. Therefore, using Lenzen's routing, in $\CLIQUE$  we can gather all information about the bad vertices to one distinguished vertex $v^\star$, and then  $v^\star$ can color them locally.
More precisely, we have the following lemma~\cite{BEPS16,FischerG17}; see \cref{Sect:proof-shattering-lemma} for the proof.

\begin{lemma}[The Shattering Lemma] \label{lem:shatter}
Let $c \geq 1$. Consider a randomized procedure that generates a subset of vertices $B \subseteq V$.
Suppose that for each $v \in V$, we have $\Prob[v \in B] \leq \Delta^{-3c}$, and this holds even if the random bits not in
$\Ninc^{c}(v)$ are determined adversarially.
Then, the following is true.
\begin{enumerate}
\item With probability $1 - n^{- \Omega(c')}$, each connected component in the graph induced by
$B$ has size at most $(c'/c) \Delta^{2c} \log_{\Delta} n$.
\item With probability $1 - O(\Delta^c) \cdot \exp(-\Omega(n \Delta^{-c}))$, the number of edges induced by $B$ is $O(n)$.
\end{enumerate}
\end{lemma}

\paragraph{Round Compression in $\CLIQUE$ and $\MPC$ by Information Gathering.}
Suppose we are given a $\tau$-round $\LOCAL$ algorithm $\mathcal{A}$ on a graph of maximum degree $\Delta$. A direct simulation of $\mathcal{A}$ on $\CLIQUE$ costs also $\tau$ rounds. However, if each vertex $v$ already knows all information in its radius-$\tau$ neighborhood, then $v$ can locally compute its output in zero rounds. In general, this amount of information can be as high as $\Theta(n^2)$, since there could be $\Theta(n^2)$ edges in the radius-$\tau$ neighborhood of $v$. 
For the case of $\Delta^{\tau} = O(n)$, it is possible to achieve an exponential speed-up in the round complexity in the $\CLIQUE$, compared to that of $\LOCAL$. In particular, in this case, each vertex $v$ can learn its radius-$\tau$ neighborhood in just $O(\log \tau)$ rounds in $\CLIQUE$. Roughly speaking, after $k$ rounds, we are able to simulate the product graph $G^{2^k}$, which is the graph where any two vertices with distance at most $2^{k}$ in graph $G$ are adjacent. This method is known as  {\em graph exponentiation}~\cite{LenzenW2010}, and it has been applied before in the design of algorithms in $\CLIQUE$ and $\MPC$ models, see e.g., ~\cite{Ghaffari2017,GhaffariU2018, ParterS18, Parter18, Andoni2018}.

\paragraph{Round Compression via Opportunistic Information Gathering.} 
Our goal is to achieve the $O(1)$ round complexity in $\CLIQUE$, so an exponential speed-up compared to the $\LOCAL$ model will not be enough. 
Consider the following ``opportunisitc'' way of simulating a $\LOCAL$ algorithm $\Algo$ in the $\CLIQUE$ model. Each vertex $u$ sends its local information (which has $O(\Delta \log n)$ bits) to each vertex $v \in V$ with some fixed probability $p = O(1/\Delta)$, independently, and it hopes that there exists a vertex $v \in V$ that gathers all the required information to calculate the outcome of $\Algo$ at $u$.
To ensure that for each $u$, there exists such a vertex $v$ w.h.p., it suffices that $p^{\Delta^{\tau}}  \gg \frac{\log n}{n}$. 
We note that a somewhat similar idea was key to the $O(1)$-round MST algorithm of~\cite{JurdzinskiN18} for $\CLIQUE$.

\cref{lem:speedup}, presented below, summarizes the criteria for this method to work; see \cref{sect:speed-up-alg} for the proof of the lemma. Denote $\InSize$ as the number of bits needed to represent the   random bits and the input for executing $\Algo$ at a vertex. Denote $\OutSize$ as the number of bits needed to represent the output of  $\Algo$ at a vertex. 
We  assume that each vertex $v$ initially knows a set $\NSp(v) \subseteq N(v)$ such that throughout the algorithm $\Algo$, each vertex $v$ only receives information from vertices in $\NSp(v)$. 
We write $\DeltaSp = \max_{v \in v} |\NSp(v)|$. Note that it is possible that $u \in \NSp(v)$ but $v \notin \NSp(u)$. In this case, during the execution of $\Algo$, all messages sent via the edge $\{u, v\}$  are from $u$ to $v$. 
Denote $\NSp^k(v)$ as the set of all vertices $u$ such that there is a path $(v = w_0, w_1, \ldots, w_{x-1} = u)$ such that $x \leq k$ and $w_i \in \NSp(w_{i-1})$ for each $i \in [1, x-1]$. Intuitively, if $\Algo$ takes $\tau$ rounds, then all information needed for vertex $v \in V$ to calculate its output is the IDs and the inputs of all vertices in $\NSp^{\tau}(v)$.

\begin{lemma}[\speedupalgo] \label{lem:speedup}
    Let $\Algo$ be a $\tau$-round $\LOCAL$ algorithm on $G = (V, E)$. There is an $O(1)$-round simulation of  $\Algo$ in  in $\CLIQUE$, given that (i) $\DeltaSp^{\tau} \log(\DeltaSp + \InSize /\log n) = O(\log n)$, (ii) $\InSize = O(n)$, and (iii) $\OutSize = O(\log n)$.
\end{lemma}

\subsection{Our New Technical Ingredients, In a Nutshell}
\label{subsec:ingredients}

The results in our paper are based on the following two novel technical ingredients, which are used in combination with the known tools mentioned above: (i) a new graph partitioning algorithm for coloring and (ii) a sparsification of the CLP coloring algorithm~\cite{ChangLP18}. 
We note that the first ingredient suffices for our $\CLIQUE$ result for graphs with maximum degree at least $\poly(\log n)$, and also for our $\MPC$ result. This ingredient is presented in \Cref{sect-high-deg-decomp}. The second ingredient, which is also more involved technically, is used for extending our $\CLIQUE$ result to graphs with smaller maximum degree, as well as for our $\LCA$ result. This ingredient is presented in \Cref{sec:sparseCLP}. Here, we provide a brief overview of these ingredients and how they get used in our results. 


\paragraph{Ingredient 1 --- Graph Partitioning for Coloring.} We provide a simple random partitioning that significantly simplifies and extends
the one in~\cite{Parter18,ParterS18}. The main change will be that, besides partitioning the vertices randomly, we also partition
the colors randomly.
In particular, this new procedure partitions the vertices and colors in a way that allows us to easily
apply CLP in a black box manner. 

Concretely, our partitioning breaks the graph as well as the
respective palettes randomly into many subgraphs $B_1, \ldots, B_k$ of maximum degree $O(\sqrt{n})$ and size $O(\sqrt{n})$, while ensuring that each vertex
in these subgraphs receives a random part of its palette with size close to the maximum degree of
the subgraph. The palettes for each part are disjoint, which allows us to color all parts in parallel. There will be one left-over subgraph $L$, with maximum degree  $\tilde{O}(\Delta^{3/4})$, as well as sufficiently large remaining palettes for each vertex in this left-over subgraph.
\begin{description}
\item[Application in $\CLIQUE$:] Since each subgraph has $O(n)$ edges, all of $B_1, \ldots, B_k$ can be colored, in parallel, in $O(1)$ rounds, using Lenzen's routing (\cref{lem:routing}). The left-over part $L$ is handled by recursion. We show that when $\Delta > \log^{4.1} n$, we are done after $O(1)$ levels of recursion.
\item[Application in Low-memory MPC:] We perform recursive calls on not only on $L$ but also on $B_1, \ldots, B_k$.  After $O(1)$ levels of recursion, the maximum degree can be made $O(n^{\beta})$, for any given constant $\beta > 0$, which enables us to run the CLP algorithm on a low memory $\MPC$.
\end{description}

We note that the previous partitioning approach~\cite{Parter18,ParterS18} is unable to reduce the maximum degree to below $\sqrt{n}$; this is a significant limitation that our partitioning overcomes.



\paragraph{Ingredient 2 --- Sparsification of the CLP Algorithm.}
In general, to calculate the output of a vertex $v$ in a $\tau$-round $\LOCAL$ algorithm $\Algo$, the output may depend on all of the $\tau$-hop neighborhood of $v$ and we may need to query  
 $\Delta^{\tau}$ vertices. To efficiently simulate $\Algo$ in $\CLIQUE$ or to transform $\Algo$ to an $\LCA$, a strategy is to ``sparsify'' the  algorithm  $\Algo$ so that the number of vertices a vertex has to explore to decide its output is sufficiently small. This notion of sparsification is a key idea behind some recent algorithms~\cite{Ghaffari2017,GhaffariU2018}. In the present paper, a key technical ingredient is providing such a sparsification for the $(\Delta+1)$ coloring algorithm of CLP~\cite{ChangLP18}.

The pre-shattering phase of the CLP algorithm~\cite{ChangLP18} consists of three parts: (i) initial coloring, (ii) dense coloring, and (iii) color bidding. Parts (i) and (ii) take $O(1)$ rounds;\footnote{In the preliminary versions (arXiv:1711.01361v1 and STOC'18) of~\cite{ChangLP18},  dense coloring takes $O(\log^\ast \Delta)$ time. This time complexity has been later improved to $O(1)$ in a revised full version of~\cite{ChangLP18} (arXiv:1711.01361v2).  
} part (iii) takes $\tau = O(\log^\ast \Delta)$ rounds.
In this paper, we sparsify the  color bidding part of the CLP algorithm.
We let each vertex $v$ sample $O(\poly \log \Delta)$ colors from its palette at the beginning of this procedure, and we show that with probability $1 - 1/\poly(\Delta)$, these colors are enough for $v$ to correctly execute the algorithm.
Based on the sampled colors, we can do an $O(1)$-round pre-processing step to let each vertex $v$ identify a subset of neighbors $\NSp(v) \subseteq N(v)$ of size $\DeltaSp = O(\poly \log \Delta)$ neighbors $\NSp(v) \subseteq N(v)$, and $v$ only needs to receive messages from neighbors in $\NSp(v)$ in the subsequent steps of the algorithm. 
\begin{description}
\item[Application in $\CLIQUE$:] For the case $\Delta = O(\poly \log n)$, the parameters $\tau = O(\log^\ast \Delta)$ and $\DeltaSp = O(\poly \log \Delta) = O(\poly (\log \log n))$ satisfy the condition for applying the opportunistic speedup lemma (\cref{lem:speedup}), and so the pre-shattering phase of the CLP algorithm can be simulated in $O(1)$ rounds in $\CLIQUE$.
\item[Application in Centralized Local Computation:] With sparsification, the pre-shattering phase of the CLP algorithm can be transformed into an $\LCA$ with $\Delta^{O(1)} \cdot \DeltaSp^\tau = \Delta^{O(1)}$ queries.
\end{description}

The recent work~\cite{AssadiCK18} on $(\Delta+1)$-coloring in $\MPC$ is also based on some form of palette sparsification, as follows. They showed that if  each vertex samples $O(\log n)$ colors uniformly at random, then w.h.p., the graph still admits a proper coloring using the sampled colors. 
Since we only need to consider the edges $\{u,v\}$ where $u$ and $v$ share a sampled color, this effectively reduces the degree to $O(\log^2 n)$. 
For an $\MPC$ algorithm with $\tilde{O}(n)$ memory per processor, the entire sparsified graph can be sent to one processor, and a coloring can be computed there, using any coloring algorithm, local or not. This sparsification is not applicable for our setting. In particular, in our sparsified CLP algorithm, we need to ensure that the coloring can be computed by a $\LOCAL$ algorithm with a small locality volume; this is because the final coloring is constructed distributedly via the opportunistic speedup lemma (\cref{lem:speedup}).

\newcommand{\epss}[1]{\epsilon_{#1}}

\newcommand{\aaaa}{\gamma}
\newcommand{\bbbbb}{\lambda}

\section{Coloring of High-degree Graphs via Graph Partitioning \label{sect-high-deg-decomp}}

In this section, we describe our graph partitioning algorithm, which is the first new technical ingredient in our results. As mentioned in \Cref{subsec:ingredients}, this ingredient on its own leads to our $\CLIQUE$ result for graphs with $\Delta=\Omega(\poly(\log n))$ and also our $\MPC$ result, as we will explain in \Cref{subsec:CongestedClique} and \Cref{subsec:MPC}, respectively. 
The algorithm will be applied recursively, but it is required that the failure probability is at most $1 - 1/\poly(n)$ in all recursive calls, where $n$ is the number of vertices in the original graph. Thus, in this section, $n$ does not refer to the number of vertices in the current subgraph $G = (V,E)$ under consideration.

\subsection{Graph Partitioning}
\label{subsec:partition}
\paragraph{The Graph Partitioning Algorithm.} The graph partitioning is parameterized by two constants $\aaaa$ and $\bbbbb$ satisfying $\aaaa \geq 2$ and $\bbbbb = \frac12 + \frac{2}{3\aaaa+2}$. Consider a graph $G = (V, E)$  with maximum degree $\Delta$. Note that $G$ is a subgraph of the $n$-vertex original graph, and so $n \geq |V|$.
Each vertex $v \in V$ has a palette $\Psi(v)$ of size $|\Psi(v)| \geq \max\{\deg_G(v), \Delta'\} + 1$, where  $\Delta' = \Delta - \Delta^{\bbbbb}$.
Denote  $G[S]$ as the subgraph induced by the vertices $S \subseteq V$. For each vertex $v \in V$, denote $\deg_S(v)$ as $|N(v) \cap S|$.
The algorithm is as follows, where we set $k = \sqrt{\Delta}$.
\begin{description}
    \item[Vertex Set:] The partition $V = B_1 \cup \dots \cup B_{k} \cup L$ is defined by the following procedure. Including each $v \in V$ to the set $L$ with probability $q = \Theta\left( \sqrt{\frac{\log n }{\Delta^{1/\Cdecomp}}}\right)$. Each remaining vertex joins one of   $B_1 , \ldots, B_{k}$ uniformly at random. Note that $\Prob[v \in B_i] = p(1-q)$, where $p = 1/k = 1/\sqrt{\Delta}$.
    \item[Palette:] Denote $C = \bigcup_{v \in V} \Psi(v)$ as the set of all colors. The partition $C = C_1 \cup \dots \cup C_{k}$ is defined by having each color $c \in C$ joins one of $C_1 , \ldots, C_{k}$ uniformly at random. Note that $\Prob[c \in C_i] = p$.
\end{description}

We require that with probability $1 - 1/\poly(n)$, the output of the partitioning algorithm satisfies the following properties, assuming that 
$\Delta = \omega(\log^{\aaaa} n)$.

\begin{description}
\item[i) Size of Each Part:] It is required that $|E(G[B_i])| = O(|V|)$, for each $i \in [k]$. Also,  it is required that  $|L| = O(q |V|) = O(\frac{\sqrt{\log n}}{\Delta^{1/\Cdecomp}}) \cdot |V|$. 


\item[ii) Available Colors in $B_i$:] For each $i\in \{1, \ldots,k\}$ and $v \in B_i$, the number of available colors in $v$ in the subgraph $B_i$ is $g_i(v) := |\Psi(v) \cap C_i|$.
It is required that $g_i(v) \geq \max\{\deg_{B_i}(v), \Delta_i - \Delta_i^{\bbbbb}\}+1$, where 
$\Delta_i := \max_{v \in B_i}\deg_{B_i}(v)$. 

\item[iii) Available Colors in $L$:] For each $v \in L$, define $g_L(v) := |\Psi(v)| - (\deg_G(v) - \deg_L(v))$. It is required that $g_L(v) \geq \max\{\deg_L(v), \Delta_L - \Delta_L^{\bbbbb}\}+1$ for each $v \in L$, where $\Delta_L := \max_{v \in L}\deg_{L}(v)$. Note that $g_L(v)$ represents a lower bound on the number of available color in $v$ {\em after } all of $B_1, \ldots, B_k$ have been colored.

\item[iv) Remaining Degrees:] The maximum degrees of $B_i$ and $L$ are $\deg_{B_i}(v)\leq \Delta_i = O(\sqrt{\Delta})$ and $\deg_{L}(v) \leq \Delta_L = O(q\Delta) =  O(\frac{\sqrt{\log n}}{\Delta^{1/\Cdecomp}}) \cdot \Delta$. For each vertex individually, we have $\deg_{B_i}(v)\leq \max\{O(\log n), O(1/\sqrt{\Delta}) \cdot \deg(v)\}$ and $\deg_{L}(v)\leq \max\{O(\log n), O(q) \cdot \deg(v)\}$.
\end{description}

Intuitively, we will use this graph partitioning in the following way. First compute the decomposition of the vertex set and the palette, and then color each $B_i$ using colors in $C_i$. 
Since $|E(G[B_i])| = O(|V|) = O(n)$, in the $\CLIQUE$ model we are able to send the entire graph $G[B_i]$ to a single distinguished vertex $v_i^\star$, and then $v_i^\star$ can compute a proper coloring of $G[B_i]$ locally. This procedure can be done in parallel for all $i$.
If $|E(G[L])| = O(n)$, then similarly we can let a vertex to compute a proper coloring of $G[L]$; otherwise we apply the graph partitioning recursively on $G[L]$, with the \emph{same} parameter $n$.

\begin{lemma}\label{lem:partition}
 Suppose $|\Psi(v)| \geq \max\{\deg_G(v), \Delta'\}+1$ with $\Delta' = \Delta - \Delta^{\bbbbb}$, and
  $|V| > \Delta = \omega(\log^{\aaaa}n)$, where $\aaaa$ and $\bbbbb$ are two constants satisfying $\aaaa \geq 2$ and $\bbbbb = \frac12 + \frac{2}{3\aaaa+2}$. 
  The two partitions $V = B_1 \cup \dots \cup B_{k} \cup L$ and  $C = \bigcup_{v \in V} \Psi(v) = C_1 \cup \dots \cup C_{k}$ satisfy the required properties, with probability $1 - 1/\poly(n)$. 
\end{lemma}
\begin{proof}
We prove that the properties i), ii), iii), and iv) hold with high probability. Note that for some of the bounds, it is straightforward to observe that they  hold in expectation.

\medskip

\noindent\textbf{i) Size of Each Part:} We first show that $|E(G[B_i])| = O(|V|)$, for each $i \in [k]$, with probability $1 - 1/\poly(n)$. 
To have  $|E(G[B_i])| = O(|V|)$, it suffices to have $\deg_{B_i}(v) = O(p \Delta)$ for each $v$, and $|B_i| = O(p |V|)$, since $p = 1/ \sqrt{\Delta}$.
Recall that we already have  $\Expect[\deg_{B_i}(v)] \leq (1-q)p \Delta < p \Delta$ and $\Expect[|B_i|] =(1-q)p|V| < p |V|$, so we only need to show that these parameters concentrate at their expected values with high probability. This can be established by a Chernoff bound, as follows. Note that we have $\epss{1} < 1$ and $\epss{2} < 1$. In particular, the inequality $\epss{1} < 1$ holds because of the assumption $\Delta = \omega(\log^{\aaaa}n) \geq \omega(\log^{2}n)$.
\begin{align*}
    \Prob[ \deg_{B_i}(v) \leq (1+\epss{1}) (1-q)p\Delta] &= 1 -  \exp(-\Omega(\epss{1}^2 (1-q)p\Delta)) = 1 - O(1/\poly(n)), \\
    \text{where }  \epss{1} &= 
    \Theta\left(\sqrt{\frac{\log n}{(1-q) p \Delta}}\right) = \Theta\left(\sqrt{\frac{\log n}{p \Delta}}\right).\\
    \\
    \Prob[ |B_i| \leq (1+\epss{2}) (1-q)p|V|] &= 1 - \exp(-\Omega(\epss{2}^2 (1-q)p|V|)) = 1 - O(1/\poly(n)), \\
    \text{where }  \epss{2} &= 
    \Theta\left(\sqrt{\frac{\log n}{(1-q) p |V|}}\right) = \Theta\left(\sqrt{\frac{\log n}{p |V|}}\right).    
\end{align*}

Next, we show the analogous results for $L$, i.e., with probability $1 - 1/\poly(n)$, both $|L|/|V|$ and $\Delta_L/\Delta$ are $O(q) = O\left(\frac{\sqrt{\log n}}{\Delta^{1/\Cdecomp}}\right)$, where $\Delta_L = \max_{v \in L}\deg_{L}(v)$.
Similarly, we already have 
 $\Expect[\deg_{L}(v)] \leq q \Delta$ and $\Expect[|L|] = q|V|$, and remember that $q = O(\frac{\sqrt{\log n}}{\Delta^{1/\Cdecomp}})$, so we only need to show that these parameters concentrate at their expected values with high probability, 
 by a Chernoff bound. 
\begin{align*}
    \Prob[ \deg_{L}(v) \leq (1+\epss{3}) q\Delta] &= 1 - \exp(-\Omega(\epss{3}^2 q\Delta)) = 
    1 - O(1/\poly(n)), \\
    \text{where }  \epss{3} &= 
    \Theta\left(\sqrt{\frac{\log n}{q \Delta}}\right).\\
    \\
    \Prob[ |L| \leq (1+\epss{4}) q|V|] &= 
    1 - \exp(-\Omega(\epss{4}^2 q|V|)) =
    1 - O(1/\poly(n)), \\    
    \text{where }  \epss{4} &= 
    \Theta\left(\sqrt{\frac{\log n}{q |V|}}\right). 
\end{align*}

Similarly, we have $\epss{3} < 1$ and $\epss{4} < 1$. In particular,  $\epss{3} < 1$ because $\Delta = \omega(\log^{\aaaa}n) \geq \omega(\log^2 n)$.

\medskip

\noindent\textbf{ii) Available Colors in $B_i$:} 
Now we analyze the number of available color for each set $B_i$.
Recall that for each $v \in B_i$, the number of available colors in $v$ in the subgraph $B_i$ is $g_i(v) := |\Psi(v) \cap C_i|$. We need to prove the following holds  with probability $1 - 1/\poly(n)$:
(i) $|\Psi(v) \cap C_i| \geq \deg_{B_i}(v) + 1$, and
(ii) $|\Psi(v) \cap C_i|  \geq  \Delta_i - \Delta_i^{\bbbbb} +1$, where $\Delta_i := \max_{v \in B_i}\deg_{B_i}(v)$. 
We will show
that with probability $1 - 1/\poly(n)$, we have $|\Psi(v) \cap C_i| \geq \Delta_i + 1$ for each $B_i$ and each $v \in B_i$, and this implies the above (i) and (ii).

Recall that $\Delta' = \Delta\left(1 - \Delta^{-(1-\bbbbb)}\right)$, $q = \Theta\left(\frac{\sqrt{\log n}}{\Delta^{1/\Cdecomp}}\right) \gg \Delta^{-(1-\bbbbb)}$,\footnote{The assumptions $\aaaa \geq 2$ and $\bbbbb = \frac12 + \frac{2}{3\aaaa+2}$ imply that $\bbbbb \in (1/2,3/4]$, and so $\Delta^{-(1-\bbbbb)} \leq \Delta^{-1/4} \ll q$.} and   $\epss{1} =   \Theta\left(\frac{\sqrt{\log n}}{\Delta^{1/\Cdecomp}}\right)$.
By selecting $q \geq 3\epss{1} = \Theta\left(\frac{\sqrt{\log n}}{\Delta^{1/\Cdecomp}}\right)$, we have 
\[
(1-\epss{1})p\Delta' = (1-\epss{1})\left(1-\Delta^{-(1-\bbbbb)}\right)p\Delta 
\geq (1+\epss{1})(1-q)p\Delta + 1.
\]
We already know that $\Delta_i \leq (1+\epss{1})(1-q)p\Delta $ with probability $1 - 1/\poly(n)$. In order to have $|\Psi(v) \cap C_i| \geq \Delta_i  + 1$, we only need to show that 
$|\Psi(v) \cap C_i| \leq (1-\epss{1})p\Delta'$ with probability $1 - 1/\poly(n)$. For the expected value, we know that $\Expect[|\Psi(v) \cap C_i|] = p |\Psi(v)| \geq p \Delta'$. By a Chernoff bound, we have
\[
\Prob[ |\Psi(v) \cap C_i| 
\leq (1-\epss{1})p\Delta'] = 1 - \exp(-\Omega(\epss{1}^2 p \Delta')) = 1 - O(1/\poly(n)).
\]

\medskip

\noindent\textbf{iii) Available Colors in $L$:}
Next, we consider the number of available colors in $L$. We show that with probability $1 - 1/\poly(n)$, for each $v \in L$, we have $g_L(v) \geq \max\{\deg_L(v), \Delta_L - \Delta_L^{\bbbbb}\}+1$, where
 $g_L(v) = |\Psi(v)| - (\deg_G(v) - \deg_L(v))$.
It is straightforward to see that $g_L(v) \geq \deg_L(v)+1$, since $g_L(v) = (|\Psi(v)| - \deg_G(v)) + \deg_L(v) \geq 1 + \deg_L(v)$. Thus, we only need to show that $g_L(v) \geq \Delta_L - \Delta_L^{\bbbbb} +1$.

In this proof, without loss of generality we assume $\deg_G(v) = |\Psi(v)| - 1 \geq \Delta'$.\footnote{If this is not the case, we can increase the degree of $v$ in a vacuous way by adding dummy neighbors to it. For
instance, we can add a clique of size $\Delta$ next to $v$ (to be simulated by $v$), remove a large enough matching from this clique
and instead connect the endpoints to $v$.}
Since $\Expect[\deg_L(v)] = q \deg_G(v) \geq q \Delta'$, by a Chernoff bound, we have
\begin{align*}
    \Prob[ \deg_{L}(v) \geq (1-\epss{3}) q \Delta'] &= 1 - \exp(-\Omega(\epss{3}^2 q \Delta')) =
    1 - O(1/\poly(n))
    \end{align*}
Remember that  $\epss{3} = 
    \Theta\left(\sqrt{\frac{\log n}{q \Delta}}\right) =    \Theta\left(\sqrt{\frac{\log n}{q \Delta'}}\right)$, and we already know that  $\epss{3} < 1$.
Using this concentration bound, the following calculation holds with probability $1 - 1/\poly(n)$. 
\begin{align*}
    g_L(v)
    &\geq (1 - \epss{3})q\Delta'\\
    &\geq q \Delta' - O\left(\sqrt{q \Delta' \log n}\right)\\
    &\geq q \Delta - q \Delta^{\bbbbb}  - O\left(\sqrt{q \Delta \log n}\right).
\end{align*}
 Combining this with $\Delta_L \leq (1+\epss{3}) q\Delta = q\Delta + O(\sqrt{q \Delta \log n})$, we obtain $g_L(v) \geq \Delta_L - q\Delta^{\bbbbb} - O(\sqrt{q \Delta \log n})$. Note that $q\Delta^{\bbbbb} + O(\sqrt{q \Delta \log n}) = o\left( (q\Delta)^{\bbbbb}\right) =  o\left(\Delta_L^{\bbbbb}\right)$,\footnote{The bound $\sqrt{q \Delta \log n} \ll (q\Delta)^{\bbbbb}$ can be derived from the assumptions $\bbbbb = \frac12 + \frac{2}{3\aaaa+2}$ and  $\Delta = \omega(\log^{\aaaa}n)$, as follows: $q\Delta = \Theta (\Delta^{\frac34} \log^{\frac12} n) = \omega(\log^{\frac34 \aaaa+\frac12}n) \implies \sqrt{q \Delta \log n} = (q \Delta)^{\frac12} \log^{1/2} n \ll (q \Delta)^{\frac12} (q \Delta)^{\frac{1}{2}\left(\frac34 \aaaa+\frac12\right)^{-1}} = (q \Delta)^{\bbbbb}$.
 } and so we finally obtain $g_L(v) \geq \Delta_L - \Delta_L^{\bbbbb} + 1$.

\medskip

\noindent\textbf{iv) Remaining Degrees:} The degree upper bounds of $\Delta_i$ and $\Delta_L$ follow immediately from the concentration bounds on $\deg_{B_i}(v)$ and $\deg_{L}(v)$ calculated in the proof of i). The bounds $\deg_{B_i}(v)\leq \max\{O(\log n), O(1/\sqrt{\Delta}) \cdot \deg(v)\}$ and $\deg_{L}(v)\leq \max\{O(\log n), O(q) \cdot \deg(v)\}$ can be derived by a straightforward application of Chernoff bound.
  \end{proof}

\subsection{Congested Clique Algorithm for High-Degree Graphs}
In this section, we show that the $(\Delta + 1)$-list coloring problem can be solved in $O(1)$ rounds in the $\CLIQUE$ model when the degrees are assumed to be sufficiently high.
The formal statement is captured in \Cref{thm:largedegree}.
First, we show that the partitioning algorithm can indeed be implemented in the $\CLIQUE$ model.
Then, we show how to color the parts resulting from the graph partitioning efficiently.
The proof of \Cref{thm:largedegree} is completed by showing that only $O(1)$ recursive applications of the partitioning are required.

\label{subsec:CongestedClique}
    \paragraph{Implementation of the Graph Partitioning.} 
    The partitions can be computed in $O(1)$ rounds on $\CLIQUE$.
 Partitioning the vertex set $V$ is straightforward, as every vertex can make the decision independently and locally, whereas it is not obvious how to partition $C$ to make all vertices agree on the same partition. Note that we can assume $|C| \leq (\Delta+1)|V|$; if $|C|$ is greater than $(\Delta+1)|V|$ initially, then we can let each vertex  decrease its palette size to $\Delta+1$ by removing some colors in its palette, and we will have $|C| \leq (\Delta+1)|V|$ after removing these colors.
 
A straightforward way of partitioning $C$ is to generate $\Theta(|C|\log n)$ random bits at a vertex $v$ locally, and then $v$ broadcasts this information to all other vertices. Note that it takes $O(\log k) = O(\log |V|) = O(\log n)$ bits to encode which part of $C_1 \cup \cdots \cup C_k$  each $c \in C$ is in.
A direct implementation of the approach cannot be done in $O(1)$ rounds, due to the message size constraint of $\CLIQUE$, as each vertex can send at most $\Theta(n \log n)$ bits in each round.

To solve this issue,  observe that it is not necessary to use total independent random bits for each $c \in C$, and $\Theta(\log n)$-wise independence suffices.
More precisely, suppose $X$ is the summation of $n$ $K$-wise independent
0-1 random variables with mean $p$, and so $\mu = \Expect[X] = np$. A Chernoff bound with $K$-wise Independence~\cite{SchmidtSS95} guarantees that
\[
\Prob[X \geq (1+q)\mu] \leq  \exp\left(-\min\{ K, q^2 \mu \}\right).
\]
In order to guarantee a failure probability of $1/ \poly(n)$ in all applications of Chernoff bound in \Cref{lem:partition}, it suffices that $K = \Theta(\log n)$. Therefore, to compute the decomposition $C = C_1 \cup \cdots \cup C_k$ with $K$-wise independent random bits, we only need $O(K \cdot \log ( |C| \log k)) = O(\log^2 n)$ total independent random bits. Broadcasting $O(\log^2 n)$ bits of information to all vertices can be done in $O(1)$ rounds via Lenzen's routing  (\Cref{lem:routing}).


\paragraph{The Algorithm of $(\Delta + 1)$-list coloring on High-degree Graphs.} We next present our $\CLIQUE$-model coloring algorithm for high-degree graphs, using the partitioning explained above.
\begin{theorem} \label{thm:largedegree}
    Suppose $\Delta = \Omega(\log^{4+\epsilon}n)$ for some constant $\epsilon > 0$.
    There is an algorithm that solves $(\Delta+1)$-list coloring in $\CLIQUE$ in $O(1)$ rounds.
\end{theorem}
\begin{proof}[Proof]
We show that a constant-depth recursive applications of  \Cref{lem:partition} suffices to give an $O(1)$-round $\CLIQUE$ $(\Delta + 1)$-list coloring algorithm for graphs with $\Delta = \Omega(\log^{4+\epsilon}n)$, for any constant $\epsilon > 0$.
Consider the graph $G = (V, E)$. First, we apply the graph partitioning algorithm of \Cref{lem:partition} to partition vertices $V$ into subsets $B_1, \ldots, B_{k}, L$ with parameter $n = |V|$, and $k = \sqrt{\Delta}$. After that, let arbitrary $k = \sqrt{\Delta}$ vertices to 
be responsible for coloring each $G[B_i]$. Each of these $k$ vertices, in parallel,  gathers all information of $G[B_i]$ from vertices $B_i$, and then computes a proper coloring of $G[B_i]$, where each vertex $v \in B_i$ uses only the palette $\Psi(v) \cap C_i$. The existence of such a proper coloring is guaranteed  by Property (ii). Using this approach, we can color all vertices in $V \setminus L$ in  $O(1)$ rounds using Lenzen's routing. Note that Property (i) guarantees that $|E(G[B_i])| = O(n)$. Finally, each vertex $v \in L$ removes the colors that have been taken by its neighbors in $V\backslash L$  from its palette $\Psi(v)$. In view of Property (iii),
after this operation, the number of available colors for each $v \in L$ is at least 
$g_L(v) \geq \max\{\deg_L(v), \Delta_L - \Delta_L^{\bbbbb}\}+1$.
Now the subgraph $G[L]$ satisfies all conditions required to apply \Cref{lem:partition}, so long as $\Delta_L = \omega(\log^{\aaaa} n)$. We will see that this condition is always met in  our application.

We then recursively apply the algorithm of the lemma on the subgraph induced by vertices $L$ \emph{with the same parameter} $n$. The recursion stops once we reach a point that $|E(G[L])| = O(n)$, and so we can apply Lenzen's routing to let one vertex $v$ gather all information of $G[L]$ and compute its proper coloring. 

Now we analyze the number of iterations needed to reach a point that $|E(G[L])| = O(n)$.
Here we use $\aaaa = 2$ and $\bbbbb = 3/4$.\footnote{We choose $\aaaa = 2$ (the smallest possible)  to minimize the degree requirement in \cref{thm:largedegree}.}
Define $V_1 = V$ and $\Delta_1 = \Delta$ as the vertex set and the maximum degree for the first iteration.
Let $V = B_1 \cup \dots \cup B_{k} \cup L$ be the outcome of the first iteration, and
define $V_2 = L$ and $\Delta_2 = \Delta_L$. Similarly, for $i > 2$, we define $V_i$ and  $\Delta_i$ based on the set $L$ in the outcome of the graph partitioning algorithm for the $(i-1)$th iteration. We have the following formulas.
\begin{align*}
    \Delta_1 &= \Delta\\
    \Delta_i &= \Delta_{i-1} \cdot O\left(\frac{\sqrt{\log n}}{\Delta_{i-1}^{1/\Cdecomp}}\right) & \text{ by Property iv) }\\
    |V_1| &= n\\
    |V_i| &= |V_{i-1}| \cdot O\left(\frac{\sqrt{\log n}}{\Delta_{i-1}^{1/\Cdecomp}}\right) & \text{ by Property i) }
 \intertext{
Let $\alpha > 0$ be chosen such that $\Delta = \Delta_1 = (\log n)^{2+\alpha}$, and assume $\alpha = \Omega(1)$ and $i = O(1)$. We can calculate the value of $\Delta_i$ and $|V_i|$ as follows.}
    \Delta_i &= O\left((\log n)^{2 + \alpha \cdot (\bbbbb)^{i-1}}\right)\\
    |V_i| &= n \cdot O\left((\log n)^{\alpha \left( (\bbbbb)^{i-1} - 1\right)}\right)
\end{align*}
Thus, given that $\alpha = \Omega(1)$ and $i = O(1)$, the condition of $\Delta_i = \omega(\log^{\aaaa} n) = \omega(\log^2 n)$ for applying \Cref{lem:partition} must be met. 

Next, we analyze the number of iterations it takes to make $\Delta_i |V_i|$ sufficiently small. In the $\CLIQUE$ model, if $\Delta_i |V_i| = O(n)$, then we are able to compute a proper coloring of $V_i$ in $O(1)$ rounds by information gathering.
%
%
%
  Let us write $\Delta = \log^{2 + \alpha}n$, where $\alpha = 2 + \beta$. The lemma statement implies that $\beta = \Omega(1)$. 
   Note that the condition for $\Delta_i |V_i| = O(n)$ can be re-written as 
   \[2 \alpha\left(1 - (\bbbbb)^{i-1}\right) \geq 2+\alpha.\]
   Combining this with $\alpha = 2 + \beta$, a simple calculation shows that this condition is met when 
   \[
   i \geq \log \left(\frac{8(\beta+2)}{3 \beta}\right) / \log \left(4/3\right).\]
 Since $\beta = \Omega(1)$, we have $\log \left(\frac{8(\beta+2)}{3 \beta}\right) / \log \left(4/3\right) = O(1)$, and so our algorithm takes only $O(1)$ iterations. In particular, when $\beta \geq 10.8$, i.e., $\Delta = \Omega(\log^{12.8}n)$, we have $\Delta_4 |V_4| = O(n)$, and so 3 iterations suffice. Since each iteration can be implemented in $\CLIQUE$ in $O(1)$ rounds, overall we get an algorithm with round complexity $O(1)$.
\end{proof}

\begin{remark}
Similar to the proof of \cref{{thm:largedegree}}, the graph partitioning algorithm also leads to an $O(1)$-round $\MPC$ coloring algorithm with $S = \widetilde{O}(n)$ memory per processor and $\widetilde{O}(m)$ total memory. This gives an simple alternate proof (with a slightly worse memory size) of the main result of~\cite{AssadiCK18} that $(\Delta+1)$-coloring can be solved with  $S = \widetilde{O}(n)$ memory per processor.
\end{remark}

\subsection{Massively Parallel Computation with Strongly Sublinear Memory}
\label{subsec:MPC}

We now show how to apply \Cref{lem:partition} as well as the CLP algorithm of \cite{ChangLP18}, as summarized in the following lemma, to prove \Cref{thm:MPCCol}. 

\begin{lemma}[\cite{ChangLP18, Parter18}]\label{lemma:CLP}
Let $G$ be an $n$-vertex graph with $m$ edges and maximum degree $\Delta$. Suppose any vertex $v$ has a palette $|\Psi(v)|$ that satisfies  $|\Psi(v)| \geq \max\left\{\deg_G(v)+1,\Delta-\Delta^{3/5}\right\}$. 
Then the list-coloring problem can be solved w.h.p.\ 
in $O(\sqrt{\log \log n})$ rounds of low-memory MPC with local memory $O(n^{\alpha})$ for an arbitrary constant $\alpha\in (0,1)$ and total memory $\widetilde{O}\left( \sum_v \deg_G(v)^2 \right)$ if  $\Delta^2 = O \left( n^{\alpha} \right)$.
\end{lemma}
The proof of \Cref{lemma:CLP} almost immediately follows from \cite{ChangLP18, Parter18}; there are only few changes that have to be made in order to turn their $\CLIQUE$~algorithm into a low-memory $\MPC$ algorithm. The details are deferred to \Cref{app:LowMemMPC}.

%

\begin{proof}[Proof of \Cref{thm:MPCCol}]
We present a recursive algorithm based on the randomized partitioning algorithm of \Cref{lem:partition}. 
If $\Delta= \poly (\log n)$ then the conditions of \Cref{lemma:CLP} are satisfied trivially; we can solve the problem in  $O(\log^*\Delta+\sqrt{\log \log n}) = O(\sqrt{\log \log n})$ rounds of low-memory MPC with total memory $\widetilde{O} ( n \cdot \Delta^2 ) = \widetilde{O}(m)$. Otherwise, we execute the following algorithm. 

\paragraph{Randomized Partitioning:}
Let $G$ be the graph that we want to color. We apply the randomized partitioning algorithm of \Cref{lem:partition} to $G$, which gives us sets $B_1, \dotsc, B_k$ and $L$, as well as color sets $C_1, \dotsc, C_k$. The goal is now to first color $B_1,\ldots,B_k$ with colors from $C_1, \ldots,C_k$, respectively. Since the colors in the sets $C_i$ are disjoint, this gives a proper coloring of $B:=\bigcup_{i=1}^k B_i$. Then, for every vertex in $L$, we remove all colors already used by neighbors in $B$ from the palettes, leaving us with a list-coloring problem of the graph induced by $L$ with maximum degree $\Delta_L$.

In the following, we first describe how to color each set $B_i$ with colors in $C_i$, and then how to solve the remaining list-coloring problem in $L$. For the parameters in \Cref{lem:partition}, we use $\aaaa = 6$ and $\bbbbb = 3/5$.\footnote{The choice  $\bbbbb = 3/5$ is to ensure that the number of available colors for each vertex in each subgraph meets the palette size constraint specified in \cref{lemma:CLP}.}

\paragraph{List-Coloring Problem in $B_i$:}
If the maximum degree $\Delta_{i}$ in $B_i$ satisfies $\Delta_i^2=O(n^{\alpha})$, then, by \Cref{lem:partition} ii), $B_i$ satisfies the conditions of \Cref{lemma:CLP}
We thus can apply the algorithm of \Cref{lemma:CLP} to $B_i$. Otherwise, we recurse on $B_i$. Note that this is possible since, by \Cref{lem:partition} ii) applied to $G$, $B_i$ satisfies the conditions of \Cref{lem:partition}. 

\paragraph{List-Coloring Problem in $L$:}
If the maximum degree $\Delta_{L}$ in $L$ satisfies $\Delta_L^2=O(n^{\alpha})$, then, by \Cref{lem:partition} iii) applied to $G$, $L$ satisfies the conditions of \Cref{lemma:CLP}. 
We thus can apply the algorithm of \Cref{lemma:CLP} to $L$. Otherwise, we recurse on $L$. Note that this is possible since by \Cref{lem:partition} iii), $L$ satisfies the conditions of \Cref{lem:partition}.
%

\paragraph{Number of Iterations:}
Since the maximum degree in $L$ reduces by a polynomial factor in every step, after at most $O(1/\alpha)$ steps, the resulting graph has maximum degree at most $O(n^{\alpha/2})$, where we satisfy the conditions of  \Cref{lemma:CLP}, and hence do not recurse further. Note that when recursing on sets $B_i$, the degree drop is even larger, and hence the same reasoning applies to bound the number of iterations.

\paragraph{Memory Requirements:}
It is obvious that the recursive partitioning of the input graph $G$ does not incur any overhead in the memory, neither local nor global. Now, let $\mathcal{H}$ be the set of all graphs $H$ on which we apply the algorithm of \Cref{lemma:CLP}. As we only apply this algorithm when the maximum degree $\Delta_H$ of $H$ is $O(n^{\alpha/2})$ or $\poly ( \log n)$, we clearly have $\Delta_H^2=O(n^{\alpha})$, so the algorithm \Cref{lemma:CLP} is guaranteed to run with local memory $O(n^{\alpha})$. 

It remains to show how to guarantee the total memory requirement of $\widetilde{O}(m)$, where $m$ is the number of edges in the input graph $G$, as promised in \Cref{thm:MPCCol}.
First, observe that due to the specifications of \Cref{lemma:CLP}, we can write the total memory requirement as
$\sum_{H\in\mathcal{H}}\sum_{v\in H} (\deg_H(v))^2$. 
First, assume that the graph $G$ has been partitioned at least three times to get to $H$.  By \Cref{lem:partition}~iv), the degree of any vertex $v$ in $H$ is either $\tilde{O}(1)$ or at most 
\[\deg_G(v) \cdot \tilde{O}\left(\Delta^{-\frac{1}{4}}\right) \cdot \tilde{O}\left(\Delta^{-\frac{1}{4} \cdot \frac{3}{4}}\right) \cdot
\tilde{O}\left(\Delta^{-\frac{1}{4} \cdot (\frac{3}{4})^2}\right) = \deg_G(v) \cdot \tilde{O}\left(\Delta^{-37/64}\right) < \tilde{O}\left(\sqrt{\deg_G(v)}\right).\]
Note that in the above calculation we assume $v$ always goes to the left-over part $L$ in all three iterations. If $v$ goes to $B_i$, then the degree shrinks faster. Remember that  we set $q = \tilde{O}(\Delta^{-1/4})$.
Hence, we require a total memory of 
\[
    \widetilde{O}\left(\sum_{H\in\mathcal{H}} \sum_{v\in H} (\deg_H (v))^2 \right) = \widetilde{O}\left( \sum_{H\in\mathcal{H}} \sum_{v\in H} \deg_G(v) \right) = \widetilde{O} \left( \sum_{v \in G} \deg_G (v) \right) = \widetilde{O}(m) \ .
\]
Note that the algorithm can be easily adapted to always perform at least three partitioning steps if $\Delta_H$ is bounded from below by a sufficiently large $\poly (\log n)$, because then the conditions of \Cref{lem:partition} are satisfied. On the other hand, if $\Delta_H = \poly (\log n)$, it is follows immediately that $\widetilde{O}\left( \sum_v (\deg_H(v))^2 \right) = \poly (\log n)=\widetilde{O}(1)$. 
Put together, we have $ \sum_{H\in\mathcal{H}}\sum_{v\in H} (\deg_H(v))^2 = \widetilde{O}(m)$.




\end{proof}

    \SetKwProg{Fn}{Function}{}{end}
    \SetKwProg{Proc}{Procedure}{}{end}
    \SetKwFor{Forsim}{for}{do simultaneously}{}
    \SetKwFor{For}{for}{do}{}
    \SetKwIF{If}{ElseIf}{Else}{if}{then}{else if}{else}{}

    \SetKwFunction{generatecolor}{GenerateColor}
    \SetKwFunction{samplecolors}{SampleColors}
    \SetKwFunction{sparsifiedcolorbidding}{SparsifiedColorBidding}
    \SetKwFunction{sparsifiedcoloring}{SparsifiedColoring}
    
    \newcommand{\MyAnd}{\textbf{\upshape and}}
    \newcommand{\MyError}{\textbf{\upshape error}}

    \newcommand{\plog}{\;\textcolor{red}{\polylog}\;}

    \newcommand{\Eirich}{E_{u_r}^{\mathrm{rich}}}
    \newcommand{\Eilazy}{E_{u_r}^{\mathrm{lazy}}}
    \newcommand{\Eilucky}{E_{u_r}^{\mathrm{lucky}}}
    \newcommand{\Eiactive}{E_{u_r}^{\mathrm{active}}}
    \newcommand{\Eicolored}{E_{u_r}^{\mathrm{colored}}}
    
\section{Distributed Coloring with Palette Sparsification \label{sec:sparseCLP}}

In this section, we present our sparsification for the $\LOCAL$-model coloring algorithm of CLP~\cite{ChangLP18}, which is the second novel technical ingredient in our results. As a consequence, this sparsification gives us (i) an $\LCA$ solving $(\Delta+1)$ list coloring  with query complexity $\Delta^{O(1)} \cdot O(\log n)$ and (ii) an $O(1)$-round $\CLIQUE$ algorithm solving $(\Delta+1)$ list coloring  for the case $\Delta = O(\poly \log n)$, using the speedup lemma (\cref{lem:speedup}).


\paragraph{The Chang-Li-Pettie Coloring Algorithm.} We will not sparsify the entire algorithm of~\cite{ChangLP18}. 
The algorithm of~\cite{ChangLP18} is based on the graph shattering framework.
Each vertex successfully colors itself with probability $1 - 1/\poly(\Delta)$ during the pre-shattering phase of~\cite{ChangLP18}, and so by the shattering lemma (\cref{lem:shatter}), the remaining uncolored vertices $V_{\bbbb}$ form connected components of size $\Delta^{O(1)} O(\poly \log n)$.\footnote{In the analysis of~\cite{ChangLP18}, this can also be made $O(\poly \log n)$, regardless of $\Delta$.} The post-shattering phase then applies a deterministic $(\deg+1)$-list coloring algorithm to color them. \cref{lem:shatter} guarantees that the number of edges within  $V_\bbbb$ is $O(n)$, and so they can be colored in $O(1)$ rounds in the $\CLIQUE$ model. Similarly, dealing with   $V_\bbbb$  only adds an $\Delta^{O(1)} \cdot O(\log n)$-factor overhead for $\LCA$. Thus, we only need to focus on the pre-shattering phase, which consists of the following three steps.

\begin{description}
    \item[Initial Coloring Step:] This step is an $O(1)$-round procedure that generates excess colors at vertices that are locally sparse.
    \item[Dense Coloring Step:] This step is an $O(1)$-round procedure that colors most of the locally dense vertices. 
    \item[Color Bidding Step:] This step is an $O(\log^\ast \Delta)$-round procedure that colors most of the remaining uncolored vertices, using the property that these vertices have large number of excess colors.  
\end{description}

For our $\LCA$ and $\CLIQUE$ algorithms, the plan is to run the initial coloring step and the dense coloring step by a direct simulation, which costs $O(1)$ rounds. 
Then, we will give a sparsified version of the color bidding step where each vertex $v$ only need to receive the information from $O(\poly \log \Delta)$ of its neighbors to decide its output.


\paragraph{A Black Box Coloring Algorithm.} In view of the above, we will use {\em part of} the algorithm of~\cite{ChangLP18} as a black box. The  specification of this black box is as follows. Consider an instance of the $(\Delta+1)$-list coloring on the graph $G=(V,E)$. The black box algorithm colors  a subset of $V$ such that the remaining uncolored vertices are partitioned into three subsets $V_{\Good}$, $V_{\bbbb}$, and $R$ meeting the following conditions.
\begin{description}
\item[Good Vertices:] The edges within $V_{\Good}$  are oriented as a DAG, and each vertex $v \in V_{\Good}$ is associated with a parameter $p_v  \leq |\Psi(v)| - \deg(v)$ satisfying the conditions $p^\star = \min_{v \in V} p_v \geq \Delta / \log \Delta$ and $\sum_{u \in  \Nout(v)} 1/ p_u \leq 1/C$, where $C > 0$ can be any specified constant.\footnote{Here $\Psi(v)$ is the set of available colors at $v$, i.e., the colors in the palette of $v$ that have not been taken by $v$'s neighbors. Here $\deg(v)$ refers to the number of uncolored neighbors of $v$ in $V_{\Good}$. We use $\outdeg(v)$ to refer to the number of out-neighbors of $v$. Intuitively, $p_v  \leq |\Psi(v)| - \deg(v)$ is a lower bound on the number of {\em excess colors} at $v$.} Recall that $\Nout(v)$ refers to the set of out-neighbors of $v$.
\item[Bad Vertices:] The probability that a vertex $v \in V$ joins $V_\bbbb$  is $1 - 1/\poly(\Delta)$. In particular, in view of \cref{lem:shatter}, with probability $1 - 1/ \poly(n)$, they form connected components of size $\Delta^{O(1)} \cdot O( \log n)$, and the number of edges within the bad vertices is $O(n)$.
\item[Remaining Vertices:] The subgraph induced by $R$ has a constant maximum degree.
\end{description}

\cref{lem:CLP-summary} follows from~\cite{ChangLP18}, after some minor modifications. 
For the sake of completeness we show the details of how we obtain \cref{lem:CLP-summary} from the results in~\cite{ChangLP18} in Appendix~\ref{sect:CLP-details}. Note that for the case of $\CLIQUE$, as long as  $\Delta = O(\sqrt{n})$, \cref{lem:CLP-summary} can be implemented in $O(1)$ rounds.

\begin{lemma}[\cite{ChangLP18}]\label{lem:CLP-summary}
Consider an instance of the $(\Delta+1)$-list coloring on the graph $G=(V,E)$.
 There is an $O(1)$-round  $\LOCAL$ algorithm that colors a subset of vertices such that the remaining uncolored vertices are partitioned into three subsets $V_{\Good}$, $V_{\bbbb}$, and $R$ meeting the above conditions, and the algorithm uses $O(\Delta^2 \log n)$-bit messages.
\end{lemma}

\subsection{A Sparsified Color Bidding Algorithm \label{sect-color-bidding-main}}
In view of \cref{lem:CLP-summary},
 we focus on the subgraph induced by $V_{\Good}$, and denote it as
$G_0=(V_0,E_0)$.  The graph $G_0$ is a directed acyclic graph.
The set of available colors for $v$ is denoted as $\Psi_0(v)$.
Our goal is to give a proper coloring of $G_0$. An important property of  $G_0$ is that
each vertex $v \in V$ is associated with a parameter $p_v \leq |\Psi_0(v)| - \deg_{G_0}(v)$ such that $\sum_{u \in  \Nout(v)} 1/ p_u \leq 1/C_0$, where $C_0$ can be any specified large constant.
Intuitively, $p_v$ gives the lower bound of the number of \emph{excess colors} at  vertex $v$.
It is guaranteed that $p^\star  = \min_{v\in V_0} p_v \geq \Delta / \log \Delta$. Parameters $C_0$ and $p^\star$ are initially known to all vertices in $V_0$.


\paragraph{Review of the Color Bidding Algorithm.}
The above conditions might look a bit strange, but it allows us to find a proper coloring in $O(\log^\ast \Delta)$ rounds in the $\LOCAL$ model by applying $O(\log^\ast \Delta)$ iterations of the procedure \trial~\cite{ChangLP18}, as follows.
\begin{enumerate}
\item Each color $c \in \Psi(v)$ is added to $S_v$ with probability $\frac{C}{2  |\Psi(v)|}$ independently.
\item If there exists a color $c^\star \in S_v$ that is not selected by any vertex in $\Nout(v)$, $v$ colors itself $c^\star$.
\end{enumerate}
We give a very high-level explanation about how this works.
For the first iteration we use $C = C_0$. 
Intuitively, for each color $c \in S_v$, the probability that $c$ is selected by an out-neighbor of $v$ is  \[\sum_{u \in  \Nout(v)} C/ (2|\Psi(u)|)  \leq \sum_{u \in  \Nout(v)} C/ (2p_u) \leq 1/2.\]
In the calculation we use the inequality  $\sum_{u \in  \Nout(v)} 1/ p_u \leq 1/C_0$ that is guaranteed by \cref{lem:CLP-summary}. The probability that $v$ fails to  color itself is roughly $1/2^{|S_v|}$, which is exponentially small in $C_0$, as in expectation $|S_v| = C_0/2$.
Thus, for the next iteration we may use a parameter $C$ that is exponentially small in $C_0$, and so after $O(\log^\ast \Delta)$ iterations, we are done.

\paragraph{Parameters.} Let $\Csamples > 0$ be a constant to be determined. Let $p^\star \in [\Delta/ \log \Delta, \Delta]$ be the parameter specified in the conditions for  \cref{lem:CLP-summary}. The $C$-parameters used in the algorithms $C_0, \ldots, C_{k-1}$ are defined as follows.
For the base case, $C_0$ is the parameter $C$ specified in the conditions for \cref{lem:CLP-summary}. Given that $C_{i}$ has been defined, we set 
\[C_{i+1} = 
2\left\lceil
\left( \min\left\{ \frac12 \exp(C_i/6)  C_i, \ \log^{\Csamples} p^\star\right\}
\right) / 2\right\rceil - 2.
\]
In other words,  $C_{i+1}$ is the result of rounding $ \min\left\{ \frac12 \exp(C_i/6)  C_i, \ \log^{\Csamples} p^\star\right\}$  down to the nearest even number.
The number of iterations $k$ is chosen as the smallest index such that $C_{k-1} = 2\left\lceil\log^{\Csamples} p^\star / 2\right\rceil - 2$. It is clear that $k = O(\log^\ast \Delta)$, as $p^\star  \leq \Delta+1$.

We will use this sequence $C_0, \ldots, C_{k-1}$ in our sparsified color bidding algorithm. This sequence is  slightly different than the one used in~\cite{ChangLP18}. The last number in the sequence used  in~\cite{ChangLP18}   is set to be $\sqrt{p^\star}$, but here we set it to be $O(\poly \log p^\star)$. Having a larger $C$-parameter leads to a smaller failure probability, but it comes at a cost that we have to sample more colors, and this means that each vertex needs to communicate with more neighbors to check for conflict.

\paragraph{Overview of the Proof.} We first review the analysis of the multiple iterations of \trial\ in~\cite{ChangLP18}, and then we discuss how we sparsify this algorithm. The proof in~\cite{ChangLP18} maintains an invariant $\Invariant_i(v)$ for each vertex $v$ that is uncolored at the beginning of each iteration $i$, as follows.\footnote{In this section, $G$ refers the current graph under consideration, i.e., it excludes all vertices that have been colored or removed in previous iterations. We use $G_0$ to refer to the original graph.}
\begin{align*}
    &\Invariant_i(v):  \sum_{u \in  \Nout_G(v)} 1/ p_u \leq 1/C_i.
\end{align*}
We will use the \emph{same} $p_u$ because the number of excess colors of a vertex never decreases. 
By  \cref{lem:CLP-summary}, this invariant is met for $i=0$. The vertices $u$ not satisfying the invariant $\Invariant_{i}(v)$ are considered {\em bad}, and are removed from consideration.
The analysis of~\cite{ChangLP18} shows that
\begin{enumerate}
    \item  Suppose  all vertices $u$ in $G$ at the beginning of the $i$th iteration satisfy the $\Invariant_{i}(u)$. Then at end of this iteration, for each vertex $u$,  with probability $1 - 1/\poly(\Delta)$, either $u$ has been successfully colored, or $\Invariant_{i+1}(u)$ is satisfied.
    \item  For the last iteration, Given that all vertices $u$ in $G$ satisfy  $\Invariant_{k-1}(u)$, then $v$ is successfully colored at iteration $k$ with probability $1 - 1/\poly(\Delta)$. 
\end{enumerate}
By the shattering lemma (\cref{lem:shatter}), all vertices that remain uncolored at the end of the algorithm induce a subgraph with $O(n)$ edges. In particular, in $\CLIQUE$ we are able to color them in $O(1)$ additional rounds.

To sparsify the algorithm, our strategy is to let each vertex sample the colors needed in all iterations at the beginning of the algorithm. It is straightforward to see that each vertex only needs to use $O(\poly \log \Delta)$ colors throughout the algorithm, with probability $1 - 1/\poly(\Delta)$. After sampling the colors, if $u$ finds that $v \in \Nout(u)$ do not share any sampled color, then there is no need for $u$ to communicate with $v$. This effectively reduces the maximum degree to $\Delta' = O(\poly \log \Delta)$. If $\Delta = O(\poly \log n)$, then  $\Delta' = O(\poly (\log \log n))$, which is enough to apply the opportunistic speedup lemma (\cref{lem:speedup}).

There is one issue needed to be overcome. That is, verifying whether $\Invariant_i(u)$ is met has to be done on the original graph $G$, as we have to go over all vertices $v \in \Nout_G(u)$, regardless of whether $u$ and $v$ have shared sampled colors.
One way to deal with this issue is to simply not remove the vertices $u$ violating $\Invariant_i(u)$, but if we do it this way, then when we calculate the failure probability of a vertex $v$, we have to apply a union bound over all vertices $u$ within radius $\tau = O(\log^\ast \Delta)$ to $v$  that $u$ does not violate the invariant for each iteration. Due to this union bound, we can only upper bound the size of the connected components of bad vertices by $\Delta^{O(\log^\ast \Delta)} \cdot O(\log n)$, so this does not lead to an improved $\LCA$.\footnote{We remark that this is only an issue for $\LCA$, and this is not an issue for application in $\CLIQUE$. In the shattering lemma (\cref{lem:shatter}), for the parameters $\Delta = O(\poly \log n)$  and  $c = O(\log^\ast \Delta)$, we can still bound the number of edges within the bad vertices $B$ by $O(n)$.}
To resolve this issue, we observe that  the invariant  $\Invariant_i(u)$  might be too strong for our purpose, since intuitively if  $v \in \Nout(u)$ does not share any  sampled colors with $u$, then $v$ \emph{should not be able to affect} $u$ throughout the algorithm. 

In this paper, we will consider an alternate invariant $\Invariant_i'(u)$ that can be checked in the sparsified graph.
More precisely, in each iteration, each vertex $v$ will do a {\em two-stage} sampling to  obtain two color sets $S_v \subseteq T_v \subseteq \Psi(v)$. The set $S_v$ has size  $C/2$, and the set $T_v$ has size $\log^{\Csamples} \Delta$, where $\Csamples >0$ is a constant to be determined. The alternate invariant $\Invariant_i'(v)$ is defined as
\[
\Invariant_i'(v): \left|T_v \setminus \bigcup_{u \in \Nout_G(v)}S_u \right| \geq |T_v|/3.\]
This invariant $\Invariant_i'(v)$ can be checked by having $v$ communicating only with its neighbors that share a sampled color with $v$.
Intuitively, if $\Invariant_i(v)$ holds, then  $\Invariant_i'(v)$ holds with probability $1 - 1/\poly(\Delta)$. It is also straightforward to see that $\Invariant_i'(v)$ implies that $v$ has a high probability of successfully coloring itself in this iteration, as $S_v$ is a size-$(C_i/2)$ uniformly random subset of $T_v$. In subsequent discussion, we say that $v$ is {\em rich} if $\Invariant_i'(v)$ is met.
Other than not satisfying $\Invariant_i'(v)$, there are two other bad events that we need to consider: 
\begin{itemize}
    \item (Informally) $v$ has too many neighbors that share a sampled color with $v$; in this case, we say that $v$ is {\em overloaded}. This is a bad event since the goal of the palette sparsification is to reduce the number of neighbors that $v$ needs to receive information from.
    \item  Most of the sampled colors of $v$ reserved for iteration $i$ have already be taken by the neighbors of $v$ during the previous iterations $1, \ldots, i-1$, so $v$ does not have enough colors to correctly run the algorithm for the $i$th iteration; in this case,  we say that $v$ is {\em lazy}.
\end{itemize}




\paragraph{The Sparsified Color Bidding Algorithm.} We are now in a position to describe the sparsified version of \trial.
For the sake of clarity we use the following notations to describe the  palette of a vertex $u$. Recall that $\Psi_0(u)$ refers to the palette of $u$ initially in the original graph $G_0$.
At the beginning of an iteration, we write $\Psi^+(u)$ to denote the set of available colors at $u$, and write $\Psi^-(u)$ to denote the set of colors already taken by vertices in $N_{G_0}(u)$. Note that $\Psi^+(u) = \Psi_0(u) \setminus \Psi^-(u)$.

The function \samplecolors\ describe how we sample the colors $S_u$ and $T_u$ in an iteration. 
Intuitively, we use $k_1 = C/2$ and $k_2 = \log^\Csamples p^\star$ as the target set sizes.
The set $\randcolor$ represents a length-$K$ sequence of colors that $u$ {\em pre-sampled} for the $i$th iteration, where $K = \log^{\Ctries+\Csamples} p^\star$, and $\randcolor(j)$ represents the $j$th color of  $\randcolor$.
We will later see that  $\randcolor$ is generated in such a way that each  $\randcolor(j)$ is a uniformly random color chosen from $\Psi_0(u)$, where $\Psi_0(u)$ is the set of available colors of $v$ initially in $G_0$. The set $S^{-}$ represents the set $\Psi^-(u)$ which consists of the colors already taken by the vertices in $N_{G_0}(u)$ before iteration $i$.

\begin{algorithm}[H]
    \Fn{\samplecolors{$k_1$, $k_2$, $S^-$, $\randcolor$}}{
        $T \leftarrow \varnothing$\;
        \For{$j \leftarrow 1$ \KwTo $k_2 \log^{\Ctries} p^\star$}{
            $c \leftarrow {\randcolor}(j)$\;
            \If{$c \notin S^-$}{
                $T \leftarrow T \cup \{c\}$\;
            }
            \If{$|T| = k_1$}{
                $T_1 \leftarrow T$\;
            }
            \If{$|T| = k_2$}{
                \KwRet{$(T_1, T)$}\; 
            }
        }
        \KwRet{$(\varnothing, \varnothing)$}\;
    }
\end{algorithm}

The procedure \sparsifiedcolorbidding\ is the sparsified version of \trial.
In this procedure, it is straightforward to verify that the outcome $S_v \leftarrow T_1$ and  $T_v \leftarrow T$ of   \samplecolors{$C/2$, $\log^{\Csamples} p^\star$, $\Psi^-(v), {\randcolor}_v$} satisfies either one of the following:
\begin{itemize}
    \item $S_v = \varnothing$ and  $T_v = \varnothing$. This happens when most of the pre-sampled colors for this iteration have been taken by the neighboring vertices. We will later show that this occurs with probability $1/\poly(\Delta)$.
    \item Given that each ${\randcolor}_v(j)$ is a uniformly random color of $\Psi_0(v)$, we have: (i) $S_v$ is a size-$(C/2)$ uniformly random subset of $T_v$, and (ii)  $T_v$ is a size-$\left(\log^{\Csamples}p^\star\right)$ uniformly random subset of $\Psi_0(v)\setminus \Psi^-(v)$. That is, these two sets $S_v$ and $T_v$ are sampled uniformly randomly from the set of available colors of $v$, i.e., $\Psi_0(v) \setminus  \Psi^-(v)$.
\end{itemize}

The condition for $v$ to be {\em overloaded} is defined in the procedure \sparsifiedcoloring. Intuitively, $v$ is said to be overloaded at iteration $i$ if the colors in  $\randcolor_v^{(i)}$ have appeared in $\bigcup_{0 \leq i' \leq i} \randcolor_u^{(i')}$, for too many neighbors $u \in N_{G_0}(v)$; this is undesirable as we want the degree of the sparsified graph to be small.

\begin{algorithm}[H]
    \Proc{\sparsifiedcolorbidding{$G$, $C$, $\Psi^-$, $\{\randcolor_v\}_{v \in V_0}$}}{
        \Forsim{each vertex $v \in V(G)$}{
            1. $(S_v, T_v) \leftarrow$ \samplecolors{$C/2$, $\log^{\Csamples} p^\star$, $\Psi^-(v), {\randcolor}_v$}. 
            If $v$ is {\em overloaded}, reset $(S_v, T_v) \leftarrow (\varnothing, \varnothing)$. We call $v$ \emph{lazy} if $S_v = \varnothing$.\\
            2. $v$ collects information about $S_u$ from all neighbors $u \in \Nout_G(v)$. \\
            3. If $\left|T_v \setminus \bigcup_{u \in \Nout_G(v)}S_u \right| \geq |T_v|/3$, i.e., at most $2/3$ of colors $v$ sampled in $T_v$ are selected in $S_u$ of some neighbors $u \in \Nout_G(v)$, then we call $v$ \emph{rich}. If (i) $v$ is  not rich or (ii) $v$ is lazy, then $v$ marks itself \badvertex\ and it skips the next step. \\
            4. If there is a color $c \in S_v$ that is not in $\bigcup_{u \in \Nout_G(v)}S_u$, we call $v$ \emph{lucky} with color $c$. If $v$ is lucky with $c$, $v$ colors itself $c$.  Tie is broken arbitrarily.
        }
    }
\end{algorithm}

The procedure \sparsifiedcoloring\ represents the entire coloring algorithm, which consists of $k = O(\log^\ast \Delta)$ iterations of \sparsifiedcolorbidding. The notation $G[U]$ refers to the subgraph induced by $U$. Note that the set $U$ does not include the vertices that are marked  \badvertex, i.e., once a vertex $v$ marked itself \badvertex, it stops attempting to color itself; but a \badvertex\ vertex might still need to provide information to other vertices in subsequent iterations.


\begin{algorithm}[H]
    \Proc{\sparsifiedcoloring{}}{
        $G \leftarrow G_0$\;
        $K \leftarrow \log^{\Ctries + \Csamples} \Delta$\;
        \For(\tcc*[h]{Obviously $k = O(\log^* p^\star - \log^* C_0) = O(\log^* \Delta)$.}){$i \leftarrow 0$ \KwTo $k-1$}{
            1. $G \leftarrow G[U]$, where $U$ consists of the yet uncolored vertices in $G$ that are not \badvertex.\\
            2. Each vertex $v \in V(G_0)$ generates a color sequence ${\randcolor}^{(i)}_v(1), \ldots, {\randcolor}^{(i)}_v(K)$ by the following rule: for each $j = 1,\ldots,K$, ${\randcolor}^{(i)}_v(j)$ is a color in $\Psi_0(v)$, chosen uniformly at random, independently.  \\
            3. Each vertex $v \in V(G)$ gathers the information about $\{{\randcolor}^{(i')}_u\}_{\substack{0 \leq i' \leq i}}$ from each neighbor $u \in N_{G_0}(v)$. \\
            4. If there exist three indices $i' \in [0, i]$, $j \in [1, K]$, and  $j' \in [1, K]$ such that ${\randcolor}^{(i)}_v(j) = {\randcolor}^{(i')}_u(j')$, we say $u \in N_{G_0}(v)$ is a \emph{significant} neighbor of $v \in V(G)$. If $v$ has more than $K^2 \log \Delta$ significant neighbors, we call $v$ \emph{overloaded}. \\
            5. Each vertex $v \in V(G)$ gathers the information about the colors that have been taken by the vertices in $N_{G_0}(v)$. Let $\Psi^-(v)$ be the set of these colors. \\
            6. Call \sparsifiedcolorbidding{$G$, $C_i$, $\Psi^-$, $\{\randcolor_v^{(i)}\}_{v \in V_0}$}. \\
        }
    }
\end{algorithm}

It is straightforward to see that \sparsifiedcoloring\ can be implemented in such a way that after an $O(1)$-round pre-processing step, each vertex $v$ is able to identify $O(\poly  \log \Delta)$ neighbors such that $v$ only need to receive information from these vertices during \sparsifiedcoloring. 
In the pre-processing step, we 
 let each vertex $v$ sample the color sequences $\randcolor_v^{(i)}$ for each $0 \leq i \leq k-1$, and let each vertex $v$ learn the set of colors sampled by $N_{G_0}(v)$. 
 Based on this information, before the first iteration begins, $v$ is able to identify at most $K^2 \log \Delta = O(\poly \log \Delta)$ neighbors of $v$ for each iteration $i$ such that $v$ is sure that $v$ does not need to receive information from all other neighbors during the $i$th iteration. See \cref{sect-implement-color-bidding} for details.
 
 For the rest of \cref{sect-color-bidding-main}, we focus on the analysis of \sparsifiedcoloring.
For each iteration $i$, recall that $\Psi^+(v) = \Psi_0(v) \setminus \Psi^-(v)$ is the set of available colors at $v$ at the beginning of this iteration.
For a vertex $v \in V_0$, and its neighbor $u \in N_{G_0}(v)$, we say that $u$ is a \emph{$c$-significant} neighbor of $v$  in iteration $i$ if $c  = {\randcolor}_u^{(i')}(j)$ for some $i' \in [1, i]$ and $j \in [1,k]$.

Consider the beginning of the $i$th iteration of the for-loop in \sparsifiedcoloring.
In the graph $G=(V,E) \leftarrow G[U]$ under consideration in this iteration, we say that a vertex $v \in V$ is \emph{$(C, D)$-honest} if the following two conditions are met.
\begin{itemize}
    \item [(i)] $\sum_{u \in \Nout_G(v)} 1/p_u \leq 1/C$.
    \item [(ii)] For each color $c \in \Psi_0(v)$, $v$ has at most $D$ $c$-significant neighbors $u \in N_{G_0}(v)$ in the previous iteration.
\end{itemize}

Clearly all vertices are $(C_0,0)$-honest in $G = G_0$ at the beginning of iteration $i = 0$.
\cref{lem:shrinkii} shows that $(C, D)$-honest vertices are well-behaved. 






\begin{lemma} \label{lem:shrinkii}
    Consider the $i$th iteration of \sparsifiedcolorbidding\ in  \sparsifiedcoloring.
    Let $U$ be the set of yet uncolored vertices after this iteration. 
    Suppose a vertex $v$ is $(C, D)$-honest, with $C \leq \log^{\beta} p^\star$ and $D \leq 2K \cdot k = O(K \log^\ast \Delta)$, at the beginning of this iteration, then
    The following holds.
    \begin{enumerate}[label=\roman*)]
        \item $\Prob[v \text{\normalfont \ does not successfully color itself}]  \leq \exp(-C/6) + \exp(-\Omega(\log^{\Csamples} \Delta))$.
        \item $\Prob[v \mathrm{\ marks\ itself\ \badvertex}] \leq \exp(-\Omega(\log^{\Csamples} \Delta))$.
        \item $\Prob[\mathrm{at \ the \ beginning \ of \ the \ next \ iteration,} \ v \in U \ \mathrm{ or } \  v \ \mathrm{ is \ not} \  (C', D')\text{-}\mathrm{honest}] \\ \leq \exp(-\Omega(\log^{\Csamples} \Delta))$, where $C' = \min\left\{ \frac12 \exp(C/6)  C, \ \log^{\Csamples} p^\star\right\}$ and $D' = D + 2K$.
    \end{enumerate}
    The probability calculation only relies on the distribution of random bits generated in $\Ninc_{G_0}^2(v)$ in this iteration, i.e., $\{{\randcolor}^{(i)}_u\}_{u \in \Ninc_{G_0}^2(v)}$. In particular, the result holds even if random bits generated outside $\Ninc_{G_0}^2(v)$ are determined adversarially.
\end{lemma}

Note that \cref{lem:shrinkii}  only relies on the assumption that the vertex $v$ under consideration is $(C,D)$-honest, and the lemma works even many of neighboring of $v$ are not $(C,D)$-honest. This is in contrast to most of the analysis of graph shattering algorithms where the analysis relies on the assumption that all vertices at the beginning of each iteration to satisfy certain invariants. 
Based on \cref{lem:shrinkii}, we show that \sparsifiedcoloring\ colors a vertex with a sufficiently high probability that enables us to apply the shattering lemma.

\begin{lemma} \label{lem:spcolor}
    The algorithm \sparsifiedcoloring\ gives a partial  coloring  of $G_0$ such that the probability that a vertex $v$ does not successfully color itself with a color in $\Psi_0(v)$ is \[O(k) \cdot \exp\left(-\Omega(\log^{\Csamples} \Delta)\right) \ll 1 /\poly(\Delta),\] and this holds even if the random bits generated outside $\Ninc^2_{G_0}(v)$ are determined adversarially. 
\end{lemma}
\begin{proof}
We consider the sequence $D_0 = 0$ and $D_{i+1} = D_i + 2K$.
Suppose the algorithm does not color a vertex $v$, then $v$ must falls into one of the following cases.
    \begin{itemize}
        \item There is an index $i \in [0, k-2]$ such that $v$ is $(C_i, D_i)$-honest at the beginning of iteration $i$, but $v$ is not $(C_{i+1}, D_{i+1})$-honest at the beginning of iteration $i+1$. By \cref{lem:shrinkii}~(iii), this occurs with probability at most $(k-1) \cdot \exp\left(-\Omega(\log^{\Csamples} \Delta)\right)$.
        \item There is an index $i \in [0, k-1]$ such that $v$ is $(C_i, D_i)$-honest at the beginning of iteration $i$, but $v$ marks itself \badvertex\ in iteration $i$.  By \cref{lem:shrinkii}~(ii), this occurs with probability at most $k \cdot \exp\left(-\Omega(\log^{\Csamples} \Delta)\right)$.
        \item For the last iteration $i = k-1$, the vertex $v$ is $(C_{k-1}. D_{k-1})$-honest at the beginning of iteration $k-1$, but $v$ does not successfully colors itself by a color in its palette. in iteration $k-1$.  By \cref{lem:shrinkii}~(iii), this occurs with probability at most $ \exp\left(-\Omega(\log^{\Csamples} \Delta)\right)$.
    \end{itemize}
 Note that our analysis only relies on the distribution of random bits generated in $\Ninc_{G_0}^2(v)$, which is guaranteed by  \cref{lem:shrinkii}. That is, even if the adversary is able to decide the random bits of vertices outside of $\Ninc_{G_0}^2(v)$ throughout the algorithm \sparsifiedcoloring, the probability that $v$ does not successfully color itself is still at most $O(k) \cdot \exp\left(-\Omega(\log^{\Csamples} \Delta)\right)$.
\end{proof}

\subsection{Analysis for the Sparsified Color Bidding Algorithm \label{sect-analysis-color-bidding}}
In this section, we prove \cref{lem:shrinkii}.
We focus on the $i$th iteration of the algorithm, where the vertex $v$ is  $(C,D)$-honest, and there is no guarantee about the $(C,D)$-honesty of all other vertices.
For this vertex $v$, we write $\Evoverloaded$, $\Evlazy$, $\Evrich$, and $\Evlucky$ to denote the event that $v$ is overloaded, lazy, rich, and lucky.
Note that a lucky vertex must be rich and not lazy, and an overloaded vertex must be lazy.
In this proof we frequently use this inequality
$\Delta + 1 \geq |\Psi_0(v)| \geq |\Psi^+(v)| \geq p_v \geq p^\star = \Omega(\Delta / \log \Delta)$.  Our analysis only considers the random bits generated by vertices within $\Ninc_{G_0}^2(v)$ in this iteration.
    
\begin{claim}\label{claim-11111}
    The probability that $v$  has more than $D'=D+2K$ $c$-significant neighbors $u \in N_{G_0}(v)$ for some color $c  \in \Psi_0(v)$ in this iteration $i$ is at most $\exp(-\Omega(\log^{\Ctries+\Csamples} \Delta))$, and this implies that $\Prob[\Evoverloaded] \leq \exp(-\Omega(\log^{\Ctries+\Csamples} \Delta))$.
\end{claim}
\begin{proof}
    Since $v$ is $(C, D)$-honest, our plan is to show that for each color $c \in \Psi_0(v)$, the number of {\em new} $c$-significant neighbor $u \in N_{G_0}(v)$   brought by the color sequences in the $i$th iteration ${\randcolor}_u^{({i})}$, 
    is  at most $2K$ with probability $1 - \exp(-\Omega(\log^{\Ctries+\Csamples} \Delta))$.
    
    Write $N_{G_0}(v) = \{u_1, \ldots, u_s\}$, and let $X_r = \boldsymbol{1} \{ \text{color $c$ appears in ${\randcolor}_{u_r}^{{(i)}}$} \}$, $Y = \sum_{1 \leq r \leq s} X_j$.
    Then $Y$ is an upper bound on the number of new $c$-significant neighbors.
    Since $v$ is $(C, D)$-honest, the total number of $c$-significant neighbors is at most $D + Y$. To prove this claim, it suffices to bound the probability of  $Y > 2K$.
    Note that $X_1, \ldots, X_s$ are independent, and 
    \[\Expect[Y] = \sum_{1 \leq r\leq s}\Expect[X_r]  \leq \sum_{1 \leq r \leq s} \left(1 - \left(1 - \frac{1}{|\Psi_0(u_r)|}\right)^K\right) \leq s \cdot \frac{K}{\Delta+1} 
    \leq \frac{K \Delta}{\Delta+1} < K.\]
    By a Chernoff bound, we have $\Prob[Y \geq 2 K] \leq \exp(-\Omega(K)) = \exp(-\Omega(\log^{\Ctries+\Csamples} \Delta))$.
    By a union bound over all $c \in \Psi_0(v)$ we are done.

    Given that $v$ has no more than $D'$ $c$-significant neighbors in this iteration for every color $c \in \Psi_0(v)$, we infer that  $v$ has at most 
      \[|{\randcolor}_u^{({i})}| D' \leq K \cdot D' = K \cdot (D + 2K) \ll K^2 \log \Delta\] significant neighbors, which implies that $v$ is not overloaded. Hence we also have $\Prob[\Evoverloaded] \leq \exp(-\Omega(\log^{\Ctries+\Csamples} \Delta))$.
  \end{proof}
  
  \begin{claim}\label{claim-22222} $
    \Prob[\Evlazy] \leq \exp(-\Omega(\log^{2+\beta} \Delta)).
    $     
  \end{claim}
   \begin{proof}

    Remember that $v$ is lazy if either (i) $v$ is overloaded, or (ii) \samplecolors gives $(\varnothing, \varnothing)$.
    In view of \cref{claim-11111}, we only need to show that with probability at most  $\exp(-\Omega(\log^{2+\beta} \Delta))$, \samplecolors gives $(\varnothing, \varnothing)$.
    
    Consider the sampling process in \samplecolors, and suppose that we are in the middle of the process, and  $T$ is the current set of colors we have obtained. Suppose $|T| = r$ currently, i.e., we have selected $r$ colors from $\Psi^+(v)$. The probability that the next color ${\randcolor}_{j}$ we consider is different from these $r$ colors in $T$ is at least
\[\frac{|\Psi^+(v)| - r}{|\Psi_0(v)|} \geq \frac{|\Psi^+(v)| - \log^{\Csamples} p^\star}{|\Psi_0(v)|} \geq \frac{\Omega(\Delta / \log \Delta)}{\Delta+1}= \Omega(1/\log \Delta).\]
Remember that $|\Psi^+(v)| \geq p_v = \Omega(\Delta / \log \Delta)$ and $\log^{\Csamples} p^\star = O(\log^{\Csamples}\Delta)$.
Also remember that \samplecolors gives $(\varnothing, \varnothing)$
if after we go over all $k_2 \log^3 p^\star = \log^{3+\beta} p^\star$ elements in the sequence ${\randcolor}$, the size of $T$ is still less than $k_2 = \log^{\beta} p^\star$. The probability that this event occurs is at most
\[
\Prob[\binomial(n',p') < t'] < \Prob[\binomial(n',p') < n'p'/2],
\]
where $n' = \log^{3+\beta} p^\star = \Theta(\log^{3+\beta}\Delta)$, $p' = \Omega(1/\log \Delta)$, and $t' =  \log^{\beta} p^\star = \Theta(\log^{\beta}\Delta) \ll n'p'$. By a Chernoff bound, this
event occurs with probability at most $\exp(-\Omega(n'p')) = \exp(-\Omega(\log^{2+\beta} \Delta))$.
   \end{proof}
    
\begin{claim}\label{claim-33333} 
    $\Prob[\overline{\Evrich}] \leq \exp(-\Omega(\log^{\Csamples} \Delta))$.
\end{claim}
    \begin{proof}
    Recall that $v$ is rich if $\left|T_v \setminus \bigcup_{u \in \Nout_G(v)}S_u \right| \geq |T_v|/3$. If $v$ is lazy, then $T_v = \varnothing$, so $v$ is automatically rich. Thus, in subsequent discussion we assume $v$ is not lazy.
 We write $\Nout_G(v) = \{u_1, \ldots, u_s\}$ and let $X_r = |T_v \cap S_{u_r}|$. 
    To prove the lemma, it suffices to  show that for $Y = \sum_{r = 1}^s X_r$, we have $\Prob[Y \geq \frac23 |T_v|] \leq \exp(-\Omega(\log^{\Csamples} p^{\star})))$.

    We consider the random variable $X_r = |T_v \cap S_{u_r}|$. For notational simplicity, we write $u = u_r$. If $u$ is lazy, then $X_r = 0$.
  Suppose $u$ is not lazy.
  The set $S_u \subseteq \Psi^+(u)$ is the result of randomly choosing distinct $C/2$ colors $c_1, \ldots, c_{C/2}$ from $\Psi^+(u)$, one by one.
For each $j \in [1, C/2]$,
define $Z_{u, j}$ as the indicator random variable that $c_j \in T_v$. Then $X_r = \sum_{j=1}^{C/2} Z_{u,j}$.
    We have the following observation. In the process, when we pick the $j$th color, regardless of the already chosen colors $c_1, \ldots, c_{j-1}$, the probability that the color picked is in $T_j$ is at most $\frac{|T_v|}{|\Psi^+(u)|-(j-1)} \leq \frac{|T_v|}{|\Psi^+(u)|-(C/2)}$. 
    Thus, we have 
    \[ \Expect[Z_{u,j}] \leq \frac{|T_v|}{|\Psi^+(u)|-(C/2)}
    \leq \frac{|T_v|}{p_u-(C/2)}
    \leq \frac{1.1|T_v|}{p_u},
    \]
    since $C/2 \leq (\log^{\beta} p^\star)/2 = O(\poly \log \Delta)$ and $p_u = \Omega(\Delta /\log \Delta)$.

    Therefore, in order to bound $Y = X_1 + \cdots + X_s$ from above, we can assume w.l.o.g.~each $X_r$ is the sum of $C/2$ i.i.d.~random variables, and each of them is a bernoulli random variable with $p = \frac{1.1|T_v|}{p_u}$, and so $Y$ is the summation of $s \cdot (C/2)$ independent 0-1 random variables.
    Since $v$ is $(C, D)$-honest, we have $\sum_{u \in  \Nout_G(v)} 1/ p_u \leq 1/C$. The expected value of $Y$ can be upper bounded as follows.
    \begin{align*}
        \Expect[Y] 
        &\leq  \frac{C}{2} \sum_{v \in \Nout_G(u)} \frac{1.1 |T_v|}{p_v} \leq \frac{1.1}{2} |T_v|.
    \end{align*}
    By a Chernoff bound, we obtain
    $$
        \Prob\left[\overline{\Evrich}\right] \leq \Prob\left[Y \geq \left(\frac{1}{1.1} \cdot \frac{4}{3}\right) \left( \frac{1.1}{2} |T_v| \right)\right] \leq \exp\left(-\Omega\left(  \frac{1.1}{2} |T_v|\right)\right) \leq \exp\left(-\Omega\left(\log^{\Csamples} p^{\star}\right)\right).
    $$
    \end{proof}
    
    Using the above three claims, we now prove the three conditions specified in \cref{lem:shrinkii}.

    \paragraph{Proof of {i)}.}
    Conditioning on $\Evrich \cap \overline{\Evlazy}$, $v$ is lucky with some color unless it fails to select \emph{any} of $|T_v|/3$ specific colors from $T_v \subseteq \Psi^+(v)$.  Remember that $\Evrich$ implies that  $\left|T_v \setminus \bigcup_{u \in \Nout_G(v)}S_u \right| \geq |T_v|/3$, and if any one of them is in $S_v$, then $v$ successfully colors itself. Also remember that $S_v$ is a size-$(C/2)$ subset of $T_v$ chosen uniformly at random.
    Thus,
    $$
        \Prob[\overline{\Evlucky} \;|\; \Evrich \cap \overline{\Evlazy}]
        \leq \frac{\binom{\frac{2}{3}|T_v|}{C/2}}{\binom{|T_v|}{C/2}} \leq \left( \frac{2}{3} \right)^{C/2} \leq \exp(-C/6).
    $$

By \cref{claim-22222} and \cref{claim-33333}, we have:
    \begin{align*}
        \Prob[\overline{\Evlucky}] &\leq \Prob[\overline{\Evlucky} \;|\; \Evrich \cap \overline{\Evlazy}] + \Prob[\Evlazy] + \Prob[\overline{\Evrich}] \\
        &\leq \exp(-C/6) +  \exp(-\Omega(\log^{\Csamples} p^\star)).
    \end{align*}
    
    \paragraph{Proof of {ii)}.}
This also follows from By \cref{claim-22222} and \cref{claim-33333}.
    \begin{align*}
        \Prob[v \mathrm{\ marks\ itself\ \badvertex}] &= \Prob[\overline{\Evrich} \cup \Evlazy] \\
        &\leq \Prob[\overline{\Evrich}] + \Prob[\Evlazy] \leq \exp(-\Omega(\log^{\Csamples} p^\star)).
    \end{align*}
    
    \paragraph{Proof of {iii)}.}
    Define  $Y$ as the summation of $1/p_u$ over all vertices $u \in \Nout_{G}(v)$ such that $u \notin U$ in the next iteration.
    We  prove that  the probabilities of (a) $Y \leq 1/C'$ and (b) $v$ has more than $D'$ $c$-significant neighbors $u \in N_{G_0}(v)$ in this iteration are both at most $\exp(-\Omega(\log^{\Csamples} \Delta))$.

    For (b), it follows from \cref{claim-11111}.
    For the rest of the proof, we deal with (a).
    Write $\Nout_{G}(v) = \{u_1, \ldots, u_s\}$.
    Consider the event $E_r^\ast = E_{u_r}^{\mathrm{lucky}} \cup \overline{E_{u_r}^{\mathrm{rich}}} \cup \Eilazy$  that $u_r$ does not join $U$ in the next iteration, i.e., $u_r$ successfully colors itself  or marks itself \badvertex.

    For each $r \in [1, s]$, define the random variable  $Z_r$ as follows. Let $Z_r = 0$ if the event $ E_{u_r}^{\mathrm{lucky}} \cup \overline{E_{u_r}^{\mathrm{rich}}} \cup \Eilazy$ occurs, and $Z_r = 1/{p_{u_r}}$ otherwise.
    Clearly we have $Y = \sum_{r=1}^s Z_r$.
   Note that 
   $\Expect[Y] \leq \exp(-C/6) \cdot (1/C)$, because \[\Prob[\overline{E_r^\ast}] = \Prob[\overline{\Eilucky} \cap \Eirich \cap \overline{\Eilazy}] \leq \Prob[\overline{\Eilucky} \; | \; \Eirich \cap \overline{\Eilazy}] \leq \exp(-C/6),\] as calculated above in the proof of Condition~(i). Since $v$ is $(C,D)$-honest, we have   $\sum_{u \in  \Nout_G(v)} 1/ p_u \leq 1/C$. Combining these two inequalities, we obtain that $\Expect[Y] \leq \exp(-C/6) \cdot (1/C)$.
   Recall that $C' = \min\left\{ \frac12 \exp(C/6)  C, \ \log^{\Csamples} p^\star\right\}$, and so $\Expect[Y] \leq 1/(2C')$.

   Next, we prove the desired concentration bound on $Y$.
    Each variable $Z_r$ is within the range $[a_r, b_r]$, where $a_r = 0$ and $b_r = 1/{p_{u_r}}$.
    We have 
    \[\sum_{r=1}^s (b_r - a_r)^2 \leq  \sum_{u \in  \Nout_G(v)} 1/ p_u^2
    \leq
     \sum_{u \in  \Nout_G(v)} 1/ (p_u \cdot p^\star)
     \leq 1/ (C p^\star).\]
    Recall $\Expect[Y] \leq 1/(2C')$. By Hoeffding's inequality, we obtain
    
    \begin{align*}
        \Prob[Y \geq 1/C']
        &\leq \exp\left(\frac{-2/(2C')^2}{\sum_{r=1}^s (b_r - a_r)^2}\right).
    \end{align*}    
    
By assumptions specified in the lemma, $(1/C')^2 = \Omega(1 / \log^{2\beta} p^\star)$ and 
$1 / \sum_{r=1}^s (b_r - a_r)^2 = \Omega(C p^\star) = \Omega( p^\star / \log^{\beta} p^\star)$. Thus, 
\[ \Prob[Y \geq 1/C'] \leq \exp(-\Omega(p^\star \log^{-3\beta} p^\star)) \leq \exp(-\Omega(\sqrt{p^\star})) \ll \exp(-\Omega(\log^{\beta} p^\star)).\]
    
There is a subtle issue regarding the applicability of Hoeffding's inequality. The variables $\{X_1, \ldots, X_k\}$ are not independent, but we argue that we are still able to apply Hoeffding's inequality. Assume that $\Nout(v) = (u_1, \ldots, u_s)$ is sorted in reverse topological order, and so for each $1 \leq a \leq s$, we have $\Nout(u_a)\cap \{u_{a}, \ldots, u_s\} = \varnothing$.
We reveal the random bits in the following manner. First of all, we reveal the set $T_u$ for all vertices $u$. Now the event regarding  whether a vertex is rich or is lazy has been determined.
Then, for $r = 1$ to $s$, we reveal the set $\{ S_u \ | \ u = u_r \text{ or } u \in \Nout(u_r)\}$. This information is enough for us to decide the outcome of $Z_r$. Note that in this process, conditioning on {\em arbitrary} outcome of $Z_1, \ldots, Z_{r-1}$ and all random bits revealed prior to revealing the set $S_{u_r}$, The probability that $\overline{E_r^\ast}$ occurs is still at most $\exp(-C/6)$.

\subsection{Implementation for the Sparsified Color Bidding Algorithm \label{sect-implement-color-bidding}}
In this section, we present an implementation of  \sparsifiedcoloring in the $\LOCAL$ model such that  after an $O(1)$-round pre-processing step, each vertex $v$ is able to identify a $O(\poly  \log \Delta)$-size subset 
$\NSp(v) \subseteq N_{G_0}(v)$ of neighboring vertices such that $v$ only needs to receive information from these vertices during \sparsifiedcoloring.



\paragraph{Fixing All Random Bits.}
Instead of having each vertex $v$ generate the color sequence ${\randcolor}_v^{(i)}$ at iteration $i$, we determined all of $\{{\randcolor}_v^{(i)}\}_{\substack{ 0 \leq i \leq k-1}}$ in the pre-processing step. 
After fixing these sequences, we can regard \sparsifiedcoloring\ as a \emph{deterministic} $\LOCAL$ algorithm, where  $\{{\randcolor}_v^{(i)}\}_{\substack{ v \in V_0,\ 0 \leq i \leq k-1}}$ can be seen as the input for the algorithm.
To gather this information, we need to use messages of $k \cdot K \cdot O(\log n) = O(\poly \log \Delta) \cdot O(\log n)$ bits, where $k =O(\log^\ast \Delta)$ is the number of iterations, and $K$ is the length of the color sequence ${\randcolor}_v^{(i)}$ for an iteration.

\paragraph{Determining the Set $\NSp(v)$.}
We show how to let each vertex $v$ determine a $O(\poly \log \Delta)$-size set $\NSp(v) \subseteq N_{G_0}(v)$ based on the following information
\[\{{\randcolor}_u^{(i)}\}_{\substack{ u \in \Ninc_{G_0}(v),\ 0 \leq i \leq k-1}}.\]
such that $v$ only needs to receive messages from $\NSp(v)$ during the execution of \sparsifiedcoloring.
 We make the following two observations.
\begin{enumerate}
    \item In order for $v$ to execute \sparsifiedcolorbidding at iteration $i$ correctly, $v$ does not need to receive information from $u \in N_{G_0}(v)$ if all colors in $\{{\randcolor}_u^{(i')}\}_{\substack{ u \in \Ninc_{G_0}(v),\ 0 \leq i' \leq i}}$ do not overlap with the colors in ${\randcolor}_u^{(i)}$. In other words, $v$ only needs information from its significant neighbors.
    \item If $v$ is overloaded at iteration $i$, then $v$ knows that it is lazy in this iteration, and so the outcome of \sparsifiedcolorbidding at iteration $i$ is that $v$  sets $S_v = T_v = \varnothing$, and $v$ marks itself \badvertex.
\end{enumerate}
The above two observations follow straightforwardly from the description of \sparsifiedcoloring.
Therefore, we can define the set  $\NSp(v)$ as follows. Add $u \in N_{G_0}(v)$ to  $\NSp(v)$ if there exists an index $i \in [0, k-1]$ such that (i) $u$ is a significant neighbor of $v$ at iteration $i$, and (ii) $v$ is not overloaded at iteration $i$.

By the definition of overloaded vertices, we know that if $v$ is not overloaded at iteration $i$, then $v$ has at most $K^2 \log \Delta = O(\poly \log \Delta)$ significant neighbors for iteration $i$. 
Thus, $|\NSp(v)| = O(\poly \log \Delta)$.
Note that the set $\NSp(v)$ can be locally calculated at $v$ during the pre-processing step. 

\paragraph{Summary.}
Algorithm \sparsifiedcoloring\ can be implemented in $\LOCAL$ in the following way.
\begin{description}
\item[Pre-Processing Step.] This step is randomized, and it takes one round and uses messages of $O(\poly \log \Delta) \cdot O(\log n)$ bits. After this step, each vertex has calculated a set  $\NSp(v)$ with $|\NSp(v)| = O(\poly \log \Delta)$.
\item[Main Steps.] This is a deterministic $O(\log^\ast \Delta)$-round procedure. During the procedure, each vertex $v$ only receives messages from  $\NSp(v)$. The  output of each vertex is a color (or a special symbol $\bot$ indicating that $v$ is uncolored), which can be represented by $\OutSize = O(\log n)$ buts. The input of each vertex consists of its color sequences for all iterations, which can be represented in $\InSize = O(\poly \log \Delta) \cdot O(\log n)$ bits.
\end{description}

 Using the above implementation of  \sparsifiedcoloring, we show that there is an $\LCA$ that solves $(\Delta+1)$-list coloring with $\Delta^{O(1)} \cdot O(\log n)$ queries.

\begin{proof}[Proof of \cref{thm:lca-main}]
Consider the following  algorithm for solving  $(\Delta+1)$-list coloring.
\begin{enumerate}
    \item Run the $O(1)$-round algorithm of \cref{lem:CLP-summary}. After that, each vertex  $v$ has four possible status: (i) $v$ has been colored, (ii) $v$ is in $V_{\Good}$,  (iii) $v$ is in $V_{\bbbb}$, or $v$ is in $R$. This can be done with $\Delta^{O(1)}$ queries.
    \item The set $R$ induces a subgraph with constant maximum degree. The $\LCA$ for $(\deg+1)$-list coloring in~\cite{FraigniaudHK16} implies that each $v \in R$ only needs $O(\log^\ast n)$ queries of vertices in $R$ to compute its color.
    \item By \cref{lem:CLP-summary}, each connected component in $V_{\bbbb}$ has size $\Delta^{O(1)} \cdot O(\log n)$. We let each vertex $v \in V_{\bbbb}$ learns the component $S$ it belongs to, and apply a deterministic algorithm to color $S$.
    \item All vertices in $V_{\Good}$ run the algorithm \sparsifiedcoloring. This adds an $(\Delta \cdot \DeltaSp^k)$-factor in the query complexity,  where  $\DeltaSp = \max_{v \in V_0} |\NSp(v)| = O(\poly \log \Delta)$, and $k = O(\log^\ast \Delta)$ is the number of iterations of \sparsifiedcoloring.  Note that in \sparsifiedcoloring, when we query a vertex $v$, the set $\NSp(v)$ can be calculated from the random bits in $N_{G_0}(v) \cup \{v\}$. 
    \item By  \cref{lem:spcolor} and \cref{lem:shatter},  the vertices left  uncolored after \sparsifiedcoloring  induces connected components of  size $\Delta^{O(1)} \cdot O(\log n)$. Similarly, we let each uncolored vertex $v$ learns the component $S$ it belongs to, and apply a deterministic algorithm to color $S$.
\end{enumerate}
By the standard procedure for converting an $\LOCAL$ algorithm to an $\LCA$, it is straightforward to implement the above algorithm as an $\LCA$ with query complexity 
\[\Delta^{O(1)} \cdot \DeltaSp^k \cdot  \left( \Delta^{O(1)} \cdot O(\log n) \right)= \Delta^{O(1)} \cdot O(\log n). \qed\]
\end{proof}

Next, we show that by applying \sparsifiedcoloring\ with the speedup lemma (\cref{lem:speedup}), we can solve $(\Delta+1)$-list coloring in $\CLIQUE$ in $O(1)$ rounds when  $\Delta = O(\poly \log n)$.

\begin{theorem} \label{thm:smalldegree}
    Suppose $\Delta = O(\poly \log n)$.
    There is an algorithm that solves $(\Delta+1)$-list coloring in $\CLIQUE$ in $O(1)$ rounds.
\end{theorem}
\begin{proof}
Recall from \cref{lem:routing} that one round in $\LOCAL$ with messages of at most $O(n \log n)$ bits can be simulated in $O(1)$ rounds in $\CLIQUE$.

The first step of the algorithm is to run the black box algorithm for \cref{lem:CLP-summary}, which takes $O(1)$ rounds in $\LOCAL$ with messages of size $O(\Delta^2 \log n) \ll O(n \log n)$. The set $R$ trivially induces  a subgraph with $O(n)$ edges. By \cref{lem:shatter}, $V_{\bbbb}$ induces a subgraph with $O(n)$ edges. We use  \cref{lem:routing} to color them in $O(1)$ rounds. 

Now we focus on $V_{\Good}$. We execute the algorithm \sparsifiedcoloring\ using the above implementation. The pre-processing step takes $O(1)$ rounds in $\LOCAL$ with messages of size $O(\poly \log \Delta) \cdot O(\log n) \ll O(n \log n)$.
For the main steps, we apply the speedup lemma (\cref{lem:speedup}) with $\tau = O(\log^\ast \Delta)$, $\OutSize = O(\log n)$, $\InSize = O(\poly \log \Delta) \cdot O(\log n)$, and $\DeltaSp = O(\poly \log \Delta)$. Since we assume  $\Delta = O(\poly \log n)$, this satisfies the criterion for \cref{lem:speedup}, and so this procedure can be executed on $\CLIQUE$ in $O(1)$ rounds. 

The algorithm \sparsifiedcoloring does not color all vertices in $V_{\Good}$. However, by \cref{lem:spcolor} and \cref{lem:shatter}, we know that these uncolored vertices induces a subgraph with $O(n)$ edges. We use  \cref{lem:routing} to color them in $O(1)$ rounds.
\end{proof}

\renewcommand{\bibfont}{\normalfont\small}

\bibliographystyle{alpha}
\bibliography{reference}

\appendix

\section{Probabilistic Tools}\label{sect:tools}

In this section we review some probabilistic tools used in this paper.

\paragraph{Chernoff Bound.} Let $X$ be the summation of $n$ independent
0-1 random variables with mean $p$.
Multiplicative Chernoff bounds give the following tail bound of $X$ with mean $\mu = np$.
\[
\Prob[X \geq (1+\delta)\mu] \leq
\begin{cases}
\exp(\frac{- \delta^2 \mu}{3})    & \text{if } \delta \in [0,1]\\
\exp(\frac{- \delta \mu}{3})    & \text{if } \delta > 1.
\end{cases}
\]
Note that these bounds hold even when $X$ is the summation of $n$ \emph{negatively correlated}
0-1 random variables~\cite{DubhashiR98,DubhashiPanconesi09} with mean $p$,
i.e., total independent is not required.
These bounds also hold when $\mu > np$ is an overestimate of $\Expect[X]$.

\paragraph{Chernoff Bound with $k$-wise Independence.}   Suppose $X$ is the summation of $n$ $k$-wise independent
0-1 random variables with mean $p$. We have  $\mu \geq \Expect[X] = np$ and the following tail bound~\cite{SchmidtSS95}. 
\[
\Prob[X \geq (1+\delta)\mu] \leq  \exp\left(-\min\{ k, \delta^2 \mu \}\right).
\]
In particular, when $k = \Omega(\delta^2 \mu)$, we obtain the same asymptotic tail bound as that of Chernoff bound with total independence.

\paragraph{Chernoff Bound with Bounded Independence.}
 Suppose $X$ is the summation of $n$  independent
0-1 random variables  with bounded dependency $d$, and let $\mu \geq \Expect[X]$, where $X = \sum_{i=1}^n X_i$. Then we have~\cite{Pemmaraju01}: \[\Prob[X \geq (1+\delta)\mu] \leq O(d) \cdot \exp(-\Omega(\delta^2 \mu / d)).\]

\paragraph{Hoeffding's Inequality.}  Consider the scenario where $X = \sum_{i=1}^n X_i$, and each $X_i$ is an independent random variable bounded by the interval $[a_i, b_i]$.
Let $\mu \geq \Expect[X]$. 
Then we have the following concentration bound~\cite{Hoeffding63}.
\[
\Prob[X \geq (1+\delta)\mu] \leq  \exp\left(\frac{-2(\delta \mu)^2}{\sum_{i=1}^n (b_i - a_i)^2}\right).
\]

\section{Proof of the Shattering Lemma \label{Sect:proof-shattering-lemma}}

In this section, we prove the shattering lemma (\cref{lem:shatter}).
Let $c \geq 1$. Consider a randomized procedure that generates a subset of vertices $B \subseteq V$.
Suppose that for each $v \in V$, we have $\Prob[v \in B] \leq \Delta^{-3c}$, and this holds even if the random bits not in
$\Ninc^{c}(v)$ are determined adversarially.
Then, the following is true.
\begin{enumerate}
\item With probability at least $1 - n^{- \Omega(c')}$, each connected component in the graph induced by
$B$ has size at most $(c'/c) \Delta^{2c} \log_{\Delta} n$.
\item With probability $1 - O(\Delta^c) \cdot \exp(-\Omega(n \Delta^{-c}))$, the number of edges induced by $B$ is $O(n)$.
\end{enumerate}

\begin{proof}
Statement~(1) is well-known; see e.g.,~\cite{BEPS16,FischerG17}. Here we provide a proof for Statement~(2).
For each edge $e = \{u,v\}$,  write $X_e$ to be the indicator random variable such that $X_e = 1$ if $u \in B$ and $v \in B$. Let $X = \sum_{e \in E} X_e$. It is clear that $\Prob[X_e] \leq 2 \Delta^{-3c}$, and so $\mu = \Expect[X] \leq n \Delta^{1 - 3c} \ll n$.
By a Chernoff bound with bounded dependence $d = 2 \Delta^{c}$ 
the probability that $X > n$ is $O(\Delta^c) \cdot \exp(-\Omega(n \Delta^{-c}))$.
\end{proof}

\section{Fast Simulation of $\LOCAL$ Algorithms in $\CLIQUE$} \label{sect:speed-up-alg}

In this section, we prove \cref{lem:speedup}.
    Let $\Algo$ be a $\tau$-round $\LOCAL$ algorithm on $G = (V, E)$. We show that there is an $O(1)$-round simulation of  $\Algo$ in  in $\CLIQUE$, given that (i) $\DeltaSp^{\tau} \log(\DeltaSp + \InSize /\log n) = O(\log n)$, (ii) $\InSize = O(n)$, and (iii) $\OutSize = O(\log n)$.

\begin{proof}
    Assume $\Algo$ is in the following canonical form.
   Each vertex first generates certain amount of local random bits, and then collects all information in its $\tau$-neighborhood. 
    The information includes not only the graph topology, but also IDs, inputs, and the random bits of these vertices.
    After gathering this information, each vertex locally computes its output based on the information it gathered. 
    
    Consider the following procedure in $\CLIQUE$ for simulating $\Algo$. In the first phase, for each ordered vertex pair $(u,v)$, with probability $p$ to be determined, $u$ sends all its local information to $v$. The local information can be encoded in $\Theta(\DeltaSp \log n + \InSize)$ bits. This includes the local input of $u$, the local random bits needed for $u$ to run $\Algo$, and the list of IDs in $\NSp(u) \cup \{u\}$.
    In the second phase, for each ordered vertex pair $(u,v)$, if $v$ has gathered all the required information to calculate the output of $\Algo$ at $u$, then $v$ sends  to $u$ the output of  $\Algo$ at $u$.

At first sight, the procedure seems to take $\omega(1)$ rounds because of the $O(\log n)$-bit message size constraint of $\CLIQUE$. However, if we set $p = \Theta\left( \frac{1}{\DeltaSp + \InSize /\log n}\right)$, the expected number of $O(\log n)$-bit messages  sent from or received by a vertex is $n p \cdot \Theta(\DeltaSp  + \InSize / \log n) = O(n)$.

More precisely, let $X_u$ be the number vertices $v \in V$ to which $u$ sends its local information in the first phase; similarly, let $Y_v$ be the number of vertices  $u \in V$ sending their local information to a $v$.
We have $\Expect[X_u] = np$, for each $u \in V$, and $\Expect[Y_v] = np$, for each $v \in V$. 
By a Chernoff bound, 
so long as $np = \Omega(\log n)$, with probability $1 - \exp(-\Omega(np)) = 1/ \poly(n)$, we have $X_u = O(np)$, for each $u \in V$, and $Y_v = O(np)$, for each $v \in V$. That is, the number of $O(\log n)$-bit messages  sent from or received by a vertex is at most $n p \cdot \Theta(\DeltaSp  + \InSize / \log n) = O(n)$, w.h.p. 

We verify that $np = \Omega(\log n)$. Condition (i) implicitly requires $\DeltaSp = O(\log n)$, and Condition (ii) requires  $\InSize = O(n)$. Therefore, $np = \Theta\left(\frac{n}{\Delta + \InSize/\log n}\right) = \Omega(\log n)$. Thus, we can route all messages in $O(1)$ rounds using Lenzen's routing (Lemma \ref{lem:routing}), and so the first phase can be done in $O(1)$ rounds.

Condition (iii) guarantees that $\OutSize = O(\log n)$, and so the messages in the second phase can be sent directly in $O(1)$ rounds.
What remains to do is to show that for each $u \in V$, w.h.p., there is a vertex $v \in V$ that receives messages from all vertices in $\NSp^{\tau}(u)$ during the first phase, and so $v$ is able to calculate the output of $u$ locally.

Denote $E_{u,v}$ as the event that $v \in V$ that receives messages from all vertices in $\NSp^{\tau}(u)$ during the first phase, and denote $E_{u}$ as the event that at least one of $\{ E_{u,v} \ | \ v \in V \}$ occurs.
 We have $\Prob[ E_{u,v} ] \geq p^{\DeltaSp^{\tau}}$, since $|\NSp^{\tau}(u)| \leq \DeltaSp^{\tau}$. Thus, $\Prob[E_u] \geq 1 -  ( 1 - p^{\DeltaSp^{\tau}} )^n$.
 
 Condition (i) guarantees that $\Delta^{\tau} \log(\DeltaSp + \InSize /\log n) = O(\log n)$. By setting $p = \epsilon / (\DeltaSp + \InSize /\log n)$ for some sufficiently small constant $\epsilon$, we haves $\DeltaSp^{\tau} \log p \geq -\frac{1}{2} \log n$, and it implies $p^{\DeltaSp^{\tau}} \geq 1/\sqrt{n}$.
 Therefore, $\Prob[E_u] \geq 1 -  ( 1 - p^{\DeltaSp^{\tau}} )^n = 1 -  \exp(-\Omega(\sqrt{n}))$.
 Thus, the simulation gives the correct output for all vertices w.h.p.
\end{proof}

We remark that the purpose of the condition $\OutSize = O(\log n)$ is only to allow the messages in the second phase to be sent directly in $O(1)$ rounds. 
With a more careful analysis and using  Lenzen's routing , 
the condition $\OutSize = O(\log n)$ can be relaxed to $\OutSize = O(n)$, though in our application we only need $\OutSize = O(\log n)$.

\section{Proof of \cref{lem:CLP-summary} \label{sect:CLP-details}}
In this section, we briefly review the algorithm of~\cite{ChangLP18} and show how to obtain \cref{lem:CLP-summary} from~\cite{ChangLP18}.
The algorithm uses a sparsity sequence  defined by
$\epsilon_1 = \Delta^{-1/10}$,
$\epsilon_i = \sqrt{\epsilon_{i-1}}$ for $i>1$, and
$\ell = \Theta(\log \log \Delta)$ is the largest index such that $\frac{1}{\epsilon_{\ell}} \geq  K$ for some sufficiently large constant $K$. The algorithm first do an $O(1)$-round procedure (initial coloring step) to color a fraction of the vertex set $V$, and denote $V^\star$ as the set of remaining uncolored vertices.
The set $V^\star$ is decomposed into $\ell+1$ subsets $(V_1,\ldots, V_{\ell}, V_{\sparse})$ according to local sparsity.
The algorithm then applies another $O(1)$-round procedure (dense coloring step) to color a fraction of vertices in $V_1 \cup \cdots \cup V_{\ell}$. 

The remaining uncolored vertices in $V^\star$ after the above procedure (initial coloring step and dense coloring step) are partitioned into three subsets: $U$, $R$, and  $V_{\bad}$.\footnote{The algorithm in~\cite{ChangLP18} for coloring layer-1 large blocks has two alternatives. Here we always use the one that puts the remaining uncolored vertices in one of $R$ or $V_{\bad}$, where each vertex is added to $V_{\bad}$ with probability $\Delta^{-\Omega(c)}$.} The set $R$ induces a constant-degree graph. The set  $V_{\bad}$ satisfies the property that each vertex is added to  $V_{\bad}$  with probability $\Delta^{-\Omega(c)}$, where $c$ can be any given constant, independent on the runtime. The  vertices in $U$ satisfy the following properties.
\begin{description}
    \item[Excess Colors:] We have $V_1 \cap U = \varnothing$. Each $v \in V_i \cap U$, with $i > 1$,  has $\Omega(\epsilon_{i-1}^2 \Delta)$ excess colors. Each $v \in V_{\sparse}\cap U$ has $\Omega(\epsilon_\ell^2 \Delta) = \Omega(\Delta)$ excess colors. The number of excess colors at a vertex $v$ is defined by the number of available colors of $v$ minus the number of uncolored neighbors of $v$.
    \item[Number of Neighbors:] For each $v \in U$, and for each $i \in [2, \ell]$, the number of uncolored neighbors of $v$ in $V_i\cap U$ is $O(\epsilon_i^5 \Delta) = O(\epsilon_{i-1}^{2.5} \Delta)$. The number of  uncolored neighbors of $v$ in $V_{\sparse} \cap U$ is of course at most $\Delta = O(\epsilon_{\ell}^{2.5} \Delta)$, since $\epsilon_{\ell}$ is a constant.
\end{description}

At this moment, the two sets  $V_{\bad}$ and $R$ satisfy the required condition specified in \cref{lem:CLP-summary}. In what follows,  we focus on $U$.

\paragraph{Orientation.}
We orient the graph induced by the uncolored vertices in $U$ as follows. For any edge $\{u,v\}$, we orient it as $(u,v)$ if one of the following is true: (i) $u \in V_{\sparse}$ but $v \notin V_{\sparse}$, (ii) $u \in V_i$ and $v \in V_j$ with $i > j$, (iii) $u$ and $v$ are within the same part in the partition $V^\star = V_1 \cup \ldots V_{\ell} \cup V_{\sparse}$ and $\ID(v) < \ID(u)$. This results in a directed acyclic graph. We write $\Nout(v)$ to denote the set of out-neighbors of $v$ in this graph.

\paragraph{Lower Bound of Excess Colors.}
In view of the above, there exist universal constants $\eta > 0$ and $C > 0$ such that the following is true. For each $i \in [2,\ell]$ and each uncolored vertex $v \in V_i \setminus V_{\bad}$, we set $p_v =  \eta \epsilon_{i-1}^2 \Delta$.
For each $v \in V_{\sparse}  \setminus V_{\bad}$, we set $p_v = \eta \epsilon_{\ell}^2 \Delta$.
By selecting a  sufficiently small $\eta$, the number $p_v$ is
always a lower bound on the number of excess colors at $v$.

\paragraph{The Number of Excess Colors is Large.} 
Recall that to color the graph quickly we need the number of excess colors to be sufficiently large with respect to out-degree.
If $v \in V_i \cap U$ with $i \geq 2$, it satisfies $|\Nout(v)| = \sum_{j=2}^{i} O(\epsilon_{j-1}^{2.5} \Delta) = O(\epsilon_{i-1}^{2.5} \Delta)$. In this case,  $p_v / |\Nout(v)| = \Omega(\epsilon_{i-1}^{-0.5})$.
If $v \in V_{\sparse} \cap U$, then of course  $|\Nout(v)| \leq \Delta = O(\epsilon_{\ell}^2 \Delta)$, since $\epsilon_{\ell}$ is a constant.  In this case,  $p_v / |\Nout(v)| = \Omega(\epsilon_{\ell}^{-0.5})$.

However, due to the high variation on the palette size in our setting, $p_v / |\Nout(v)|$ is not a good measurement for the gap between the number of excess colors and out-degree at $v$. The inverse of the expression $\sum_{u \in  \Nout(v)} 1 / p_u$ turns out to be a better measurement, as it takes into account the number of excess colors in each out-neighbor.

There is a constant $C > 0$ such that for each uncolored vertex $v \in V^\star \setminus (V_{\bad} \cup R)$, we have $\sum_{u \in  \Nout(v)} 1 / p_u \leq  1/C$.  The calculation is as follows.
\begin{align*}
&\text{If $v \in V_i  \cap U$ \ $(i>1)$,} \ & & \text{then }
\sum_{u \in  \Nout(v)} 1 / p_u 
=
\sum_{j=2}^{i} O\paren{\frac{\epsilon_{j-1}^{2.5} \Delta}{\epsilon_{j-1}^2 \Delta}} = \sum_{j=2}^{i} O({\epsilon_{j-1}^{0.5}}) = O(\epsilon_{i-1}^{0.5}) < 1/C.\\
&\text{If $v \in V_{\sparse} \cap U$, } \ &  & \text{then }
\sum_{u \in  \Nout(v)} 1 / p_u 
=
\sum_{j=2}^{\ell+1} O\paren{\frac{\epsilon_{j-1}^{2.5} \Delta}{\epsilon_{j-1}^2 \Delta}} = \sum_{j=2}^{\ell+1} O({\epsilon_{j-1}^{0.5}}) = O(\epsilon_{\ell}^{0.5}) < 1/C.
\end{align*}

For a specific example, if $v$ is an uncolored vertex in $V_2 \setminus V_{\bad}$, then $p_v = \eta \epsilon_1^2 \Delta = \eta \Delta^{0.8}$ is the lower bound on the number of excess colors at $v$, and $v$ has out-degree $|\Nout(v)| = O(\epsilon_1^{2.5} \Delta) =  O(\Delta^{0.75})$, and we have $\sum_{u \in  \Nout(v)} 1 / p_u = O(\epsilon_1^{0.5}) =  O(\Delta^{-0.05}) < 1/C$. Intuitively, this means that the gap between the number of excess colors and the out-degree at $v$ is $\Omega(\Delta^{0.05})$.

\paragraph{Summary.}
Currently the graph induced by  $U$ satisfies the following conditions. Each vertex $v$ is associated with a parameter $p_v = \eta \epsilon_j^2 \Delta$ (for some $j \in [1,\ell]$) such that the number of excess colors at $v$ is at least $p_v = \Omega(\epsilon_j^2 \Delta)$, but the number of out-neighbors of $v$ is at most $O(\epsilon_j^{2.5} \Delta)$. In particular, we always have $\sum_{u \in  \Nout(v)} 1/p_u = O(\epsilon_{j}^{0.5}) < 1/C$, where $C > 0$ is a universal constant. 
The current $p_v$-values for vertices in $U$ almost satisfy the required condition for $V_{\Good}$ specified in \cref{lem:CLP-summary}. 
\begin{description}
    \item[Lower Bound of $p^\star$.] Define   $p^\star$ as the minimum $p_v$-value among all uncolored vertices $v \in V^\star$. Currently we only have  $p^\star \geq \eta \epsilon_1^2 \Delta = \eta \Delta^{0.8}$, but in \cref{lem:CLP-summary} it is required that $p^\star \geq \Delta / \log \Delta$.
    \item[Lower Bound of $C$.] Currently we have $\sum_{u \in  \Nout(v)} 1/ p_u \leq 1/C$ for some universal constant $C$, but in \cref{lem:CLP-summary} it is required that $C > 0$ can be any given constant.
\end{description}

For the rest of the section, we show that there is an $O(1)$-round that is able to improve the lower bound of $p^\star$ to $p^\star \geq \Delta / \log \Delta$ and increase the parameter $C$ to any specified constant. The procedure will colors a fraction of vertices in $U$ and puts some vertices in $U$ to the set $V_{\bad}$.
We first consider improving the lower bound of $p^\star$. This is done by letting all vertices whose $p_v$-value are too small (i.e., less than $\Delta / \log \Delta$) to jointly run  \cref{lem:color-remain-copy}. For these vertices, we have 
$\sum_{u \in  \Nout(v)} 1/ p_u \leq 
O\left(\log^{-1/4} \Delta\right)$,\footnote{Each vertex $v$  is associated with a parameter $p_v = \eta \epsilon_j^2 \Delta$, and we have $\sum_{u \in  \Nout(v)} 1/p_u = O(\epsilon_{j}^{0.5}) = O(p_v^{1/4})$.} and so we can use $C = \Omega\left( \log^{1/4} \Delta \right)$ in \cref{lem:color-remain-copy}. The algorithm of  \cref{lem:color-remain-copy} takes only $O(1)$ rounds. All participating vertices  that still remain uncolored join $V_{\bad}$.

\begin{lemma}[\cite{ChangLP18}]\label{lem:color-remain-copy}
Consider a directed acyclic graph, where vertex $v$ is associated with a parameter $p_v \leq |\Psi(v)| -  \deg(v)$
We write $p^\star = \min_{v\in V} p_v$.
Suppose that there is a number $C = \Omega(1)$ such that all vertices $v$  satisfy $\sum_{u \in  \Nout(v)} 1/ p_u \leq 1/C$.
Let $d^\star$ be the maximum out-degree of the graph.
There is an $O(\log^\ast (p^\star) - \log^\ast (C))$-time algorithm achieving the following.
Each vertex $v$ remains uncolored with probability at most $\exp(-\Omega(\sqrt{p^\star})) + d^\star \exp(-\Omega(p^\star))$.
This is true even if the random bits generated outside a constant radius around $v$ are determined adversarially.
\end{lemma}

Now the lower bound on $p^\star$ is met. We show how to increase the $C$-value to any given constant we like in $O(1)$ rounds. We apply \cref{lem:shrink-copy} using the current  $p^\star$ and $C$.
After that, we can set the new $C$-value to be $C' = C \cdot \exp(C/6) / (1+\lambda)$, after putting each vertex $v$ not meeting the following condition to $V_{\bad}$:
\[
\text{Sum of $1/p_u$ over all remaining uncolored vertices $u$ in $\Nout(v)$ is at most $1/C' = \frac{1+\lambda}{ \exp(C/6) C }$.}
\]
If $\lambda$ is chosen as a small enough constant, we have $C' > C$. After a constant number of iterations, we can increase the $C$-value to any constant we like. Now, all conditions in \cref{lem:CLP-summary} are met for the three sets $R$, $V_{\bad}$, and $V_{\Good} \leftarrow U$.
\begin{lemma}[\cite{ChangLP18}]\label{lem:shrink-copy}
There is an one-round algorithm meeting the following conditions.
Let $v$ be any vertex. Let $d$ be the summation of $1/p_u$ over all vertices $u$ in $\Nout(v)$ that remain uncolored after the algorithm.
Then the following holds.
\begin{align*}
\Prob\left[d \geq \frac{1+\lambda}{ \exp(C/6) C }\right] &\leq \exp\left(-2 \lambda^2   p^\star \exp(-C/3) / C \right) +d^\star\exp(-\Omega(p^\star)).
\end{align*}
\end{lemma}



\section{The CLP Algorithm in the Low-Memory $\MPC$ Model}\label{app:LowMemMPC}
In this section, we show which changes have to be made to \cite{ChangLP18} to get a low-memory $\MPC$ algorithm, thus proving  \Cref{lemma:CLP}.

There are two main issues in the low-memory $\MPC$ model that we need to take care of.
First, the total memory of the system is limited to $\tilde{\Theta}(m + n)$, where $m$ and $n$ are the number of edges and vertices in the input graph, respectively.
Second, the local memory per machine is restricted to $O(n^{\alpha})$, for an arbitrary constant $\alpha>0$.
These two restrictions force us to be careful about the amount of information sent between the machines.
In particular, no vertex can receive messages from more than $O(n^{\alpha})$ other vertices in one round (as opposed to the $\CLIQUE$, where a vertex can receive up to $O(n)$ messages per round).

The key feature of our partitioning algorithm is that we can reduce the coloring problem to several instances of coloring graphs with maximum degree $\Delta = O(n^{\alpha / 2})$.
Given this assumption, we can implement the CLP algorithm in the low-memory $\MPC$ model almost line by line as done by Parter~\cite[Appendix A.2]{Parter18} for the $\CLIQUE$.
Therefore, here we simply point out the differences in the algorithm and refer the reader to the paper by Parter for further technical details.

\paragraph{Dense Vertices.}
Put briefly, a vertex is $\gamma$-dense, if a $(1 - \gamma)$-fraction of the edges incident on it belong to at least $(1 - \gamma) \cdot \Delta$ triangles.
An $\gamma$-almost clique is a connected component of $\gamma$-dense vertices that have at most $\gamma \cdot \Delta$ vertices outside the component.
Each such component has a weak diameter of at most $2$.
These components can be computed in $2$ rounds by each vertex learning its $2$-hop neighborhood.
This process is performed $O(\log \log \Delta)$ times in parallel which incurs a factor of $O(\log \log \Delta)$ in the memory requirements, which is negligible.
Furthermore, the algorithm requires running a coloring algorithm within the dense components.
Since the component size is at most $\Delta \ll \Delta^2$, we can choose one vertex in the component as a leader and the leader vertex can locally simulate the coloring algorithm without breaking the local memory restriction.

\paragraph{Memory Bounds.}
Once the $2$-hop neighborhoods of nodes have been learned, no more memory overhead is required.
Since we have $\Delta \ll  n^{\alpha/2}$, learning the $2$-hop neighborhoods does not violate the local memory restriction of $O(n^{\alpha})$.
For the total memory bound, storing the $2$-hop neighborhoods requires at most $\widetilde{O}(\sum_v (\deg_G (v))^2)$ memory.


\paragraph{Post-Shattering and Clean-up.}
Another step that we cannot use as a black box is a subroutine that colors a graph that consists of connected components of $O(\poly \log n)$ size.
Regardless of the component sizes being small, all vertices over all components might not fit the memory of a single machine.
Hence, similarly to the CLP algorithm in the $\LOCAL$ model, we use the best deterministic list coloring algorithm to color the components.
For general graphs, currently the best runtime in the $\LOCAL$ model is obtained by applying the algorithm by Panconesi and Srinivasan~\cite{PanconesiS96} with runtime of $ 2^{O(\sqrt{\log n'})} $, where $n' = O(\poly \log n)$ is the maximum size of the small components.
We can improve this bound exponentially in the $\MPC$ model by using the known \emph{graph exponentiation} technique~\cite{Lenzen2013,brandt2018Arb} and obtain a runtime of $O(\sqrt{\log \log n})$.

The graph exponentiation technique works as follows.
Suppose that every vertex knows all the vertices and the topology of its $2^{i - 1}$-hop neighborhood in round $i - 1$ for some integer $i \geq 0$.
Then, in round $i$, every vertex can communicate the topology of its $2^{i - 1}$-hop neighborhood to all the vertices in its $2^{i - 1}$-hop neighborhood.
This way, every vertex learns its $2^{i}$-hop neighborhood in round $i$ and hence, every vertex can simulate any $2^i$-round $\LOCAL$ algorithm in $i$ rounds.
We observe that, in the components of $O(\poly \log n)$ size, the $2^{i}$-hop neighborhood of any vertex for any $i$ fits into the memory of a single machine since the number of vertices in the neighborhood is clearly bounded by $O(\poly \log n)$.
The same observation yields that the total memory of $\tilde{O}(m)$ suffices.

\end{document}